\documentclass[submission,copyright,creativecommons]{eptcs}
\providecommand{\event}{NCMA 2026}

\usepackage{underscore}
\usepackage[T1]{fontenc}

\usepackage{placeins}

\usepackage{listings}

\usepackage{amsmath,amssymb,amsthm}
\usepackage{algorithm}
\usepackage[noend]{algpseudocode}
\usepackage{xcolor}
\usepackage{xspace}
\usepackage{verbatim}
\usepackage{booktabs}
\usepackage{graphicx}
\usepackage{subcaption}
\usepackage{thmtools}
\usepackage{thm-restate}

\usepackage{tikz}
\usetikzlibrary {automata,positioning} 
\usetikzlibrary{calc}
\usepackage{tikzpagenodes}

\newcommand{\eps}{\ensuremath{\varepsilon}}
\newcommand{\thom}{\ensuremath{\mathit{Th}}}
\newcommand{\glush}{\ensuremath{\mathit{Gl}}}

\newcommand{\IN}[1][O]{\ensuremath{In(#1)}\xspace}

\newcommand{\entry}{\ensuremath{\mathit{entry}}}

\DeclareRobustCommand{\MFN}{\ifmmode\text{\texttt{MFN}}\else\texttt{MFN}\fi}
\DeclareRobustCommand{\CN}{\ifmmode\text{\texttt{CN}}\else\texttt{CN}\fi}
\DeclareRobustCommand{\IN}{\ifmmode\text{\texttt{IN}}\else\texttt{IN}\fi}
\DeclareRobustCommand{\I}{\ifmmode\text{\texttt{IAR}}\else\texttt{IAR}\fi}
\DeclareRobustCommand{\IR+}{\ifmmode\text{\texttt{IAR+}}\else\texttt{IAR+}\fi}
\DeclareRobustCommand{\None}{\ifmmode\text{\texttt{None}}\else\texttt{None}\fi}
\DeclareRobustCommand{\All}{\ifmmode\text{\texttt{All}}\else\texttt{All}\fi}
\newcommand{\Or}{\mathbin{|}}

\newcommand{\changeB}[1]{\textcolor{black}{#1}}

\tikzset{>=latex}

\newtheorem{theorem}{Theorem}

\newtheorem*{corollary}{Corollary}
\theoremstyle{definition}
\newtheorem{definition}{Definition}
\theoremstyle{remark}
\newtheorem{remark}{Remark}
\newtheorem{conjecture}{Conjecture}
\newtheorem{example}{Example}
\newtheorem{observation}{Observation}

\title{Selective Memoization for Efficient Backtracking Regular Expression Matching}
\author{Martin Berglund
\institute{Umeå University,\\Umeå, Sweden}
\email{mbe@cs.umu.se}
\and
Brink van der Merwe
\institute{
Stellenbosch University,\\Stellenbosch, South Africa}
\email{abvdm@cs.sun.ac.za}
\and
Iain le Roux\footnote{Authors are listed alphabetically. While this work represents a collaborative effort, Iain le Roux served as the lead author.}
\institute{
Stellenbosch University,\\Stellenbosch, South Africa}
\email{iainleroux@gmail.com}
}

\def\authorrunning{M.\ Berglund, B.\ van der Merwe, I.\ le Roux}
\def\titlerunning{Selective Memoization for Regex Matching}

\begin{document}
\maketitle

\begin{abstract}
Backtracking regular expression matchers are widely used due to their expressive power but may exhibit exponential worst-case matching time. Memoization provides a principled method for eliminating redundant computation and ensuring linear matching time, but full memoization is memory-intensive and impractical.
We introduce the Minimum Feedback Node (\MFN{}) memoization scheme, a selective memoization strategy based on computing a minimum feedback vertex set of an automaton. 
We establish relationships with existing memoization schemes and analyze their behaviour under both Thompson and Glushkov automaton constructions. 
\end{abstract}

\section{Introduction}
Regular expressions (regexes) are among the most widely used tools in
software engineering, appearing in text processing, input validation,
security filtering, and compiler construction~\cite{davis2019aren}.
Most practical regex engines -- including those in Java, Python, Ruby, and JavaScript -- use a backtracking algorithm that performs a depth-first
search through a nondeterministic finite automaton (NFA) constructed from
the regex. While this approach supports expressive features such as
capture groups and backreferences, it is vulnerable to \emph{Regular
Expression Denial of Service} (ReDoS) attacks~\cite{davis2021using}: for many regexes it is possible to construct input strings that force the matcher into exponential
worst-case matching time (in the length of the input strings).

Catastrophic backtracking is closely related to \emph{ambiguity} in the
underlying NFA --- that is, the existence of multiple distinct
computations for the same input string. When the NFA exhibits
\emph{infinite degree of ambiguity} (IDA), the backtracking matcher
explores, for some input strings, a super-linear number of transitions  in the length of the input string~\cite{weideman2016analyzing}, leading to super-linear matching 
time. 
\emph{Full memoization} prevents this by recording and reusing the result of the computation starting from each configuration (that is, each state and input position) reached, reusing this information when encountering the configuration again. Davis et al.~\cite{davis2021using}
pointed out that full memoization guarantees linear matching time, but at a potentially prohibitive memory cost.
To reduce memory overhead, \emph{selective} memoization schemes memoize
only a subset of NFA states, aiming to eliminate IDA while minimising
memory usage. Davis et al.~~\cite{davis2021using} and Van der Merwe et al.~\cite{van2021memoized} formalized
selective memoization and defined \emph{memoized NFAs} (mNFAs). The latter proposed three schemes for selecting states to memoize, each an overapproximation, as the work also proved NP-completeness of
finding a \emph{minimum} memoization set for general NFAs.

In this paper, we introduce the further refined \emph{Minimum Feedback Node} (\MFN{})
memoization scheme. \MFN{} memoizes a minimum feedback vertex set of the NFA graph (in the absence of $\varepsilon$-cycles): a smallest set of nodes whose removal renders
the NFA graph acyclic. 
In the presence of $\varepsilon$-cycles, fewer nodes are memoized than what is required for a minimum feedback vertex set of the NFA, but we deal with $\varepsilon$-cycles separately, see the definition of runs in Section~\ref{sec:background}, and Observation~\ref{obs:no-repeat-transition}.
Since every non-$\varepsilon$ cycle must contain at least one memoized
state, \MFN{} eliminates IDA. 
We present linear-time
\emph{parse tree algorithms} (PTAs) that compute a minimum feedback vertex set for both the
Thompson and Glushkov NFA constructions directly from the parse tree of
the regex, without first building the full automaton. For Thompson NFAs,
reducibility~\cite{giammarresi2004thompson} allows us to apply Shamir's
classical linear-time minimum cutset
algorithm~\cite{shamir1979linear}; for Glushkov NFAs -- which are not always reducible~\cite{GiammarresiPontyWood1998} -- we provide a new parse-tree-based algorithm.
Operating on parse trees has several advantages, as practical
regex compilers typically build a parse tree before constructing the automaton, and often leave the automaton implicit. As such, even though Shamir's algorithm would be applicable to the automaton, the algorithm for the Thompson case operates directly on the parse tree, allowing \MFN{} memoization to be computed in a single pass during compilation, without materializing the NFA. This reduces implementation complexity and avoids potential overheads. %
We further investigate the relationships between the memoization schemes of~\cite{van2021memoized} and $\MFN$, establishing that $\MFN$ always forms a subset of the key \emph{closure node} memoization, and investigate the non-minimality of a memoization scheme (\I{}) proposed in~\cite{van2021memoized}, even for automata derived using the Glushkov construction.
We also explore the differences in memoization behaviour
between Thompson and Glushkov automata, and offer extensive experimental results.

\section{Definitions and Background}
\label{sec:background}

For an alphabet (a finite set of symbols) $\Sigma$, let $\Sigma^*$ denote the set of all strings (sequences of symbols) from $\Sigma$. Let $\varepsilon$ denote the empty string. For $u,v\in \Sigma^*$, let $u\cdot v$ denote the concatenation of $u$ and $v$. When it can cause no confusion, we write simply $uv$. For any string $w$ and set $S$, let $\pi_S(w)$ denote the string resulting from removing all symbols that are not in $S$ from $w$. \emph{Regular expressions over $\Sigma$} are defined inductively:
$E \;::=\; \varepsilon \;\bigm|\; a \;\bigm|\;
  (E \Or E) \;\bigm|\; (E \cdot E) \;\bigm|\; E^{*} \;\bigm|\; E^{+} \;\bigm|\; E^?$, 
where $a \in \Sigma$, concatenation $E \cdot E'$ is often written $EE'$. 
We treat $E^+$ as a primitive operator with its own NFA construction (analogous to, but distinct from, $E^*$), rather than rewriting it as $EE^*$. 
This avoids introducing duplicate subgraph structures. 
Keeping with the goal of avoiding unnecessary states, we also treat $E^?$ as a primitive operator rather than as syntactic sugar for $E \Or \varepsilon$. We construct the NFA for $E^?$ directly from that of $E$: either by adding $\varepsilon$-transitions from the initial state to all final states or by making the initial state accepting. This avoids introducing superfluous structure and, in constructions such as Glushkov NFAs, can avoid $\varepsilon$-transitions altogether.
The language matched by a regex $E$ is defined as usual and is denoted $\mathcal{L}(E)$.

We closely follow the definitions of~\cite{davis2021using} and~\cite{van2021memoized} for the relevant automata and ambiguity concepts. A memoized non-deterministic finite automaton (mNFA) is a tuple $A=(\mathcal{M},Q,\Sigma,\delta,q_0,F)$ where: (i) $Q$ is a finite set of states; (ii) $\mathcal{M}\subseteq Q$ is the \emph{memoized} states; (iii) $\Sigma$ is the input alphabet; (iv) $q_0 \in Q$ is the initial state; (v) $\delta \subseteq Q\times (\Sigma \cup \{\varepsilon\}) \times Q$ is the transition relation; and (vi) $F\subseteq Q$ are the final states.

The set of \emph{runs of $A$ on $w\in \Sigma^*$} is the set of root to leaf paths in an unordered tree denoted $T_A(w)$ with nodes having labels from $Q\times \{0,\ldots,|w|\}$. The root is $(q_0,0)$, and a node $(q,i)$ has $(q',j)$ as a child if: (i) $(q,\alpha,q')\in \delta$; (ii) either $j=i$ and $\alpha=\eps$, or $j=i+1$ and the $j$th symbol in $w$ is $\alpha$; and; (iii) if a node labeled $(q,i)$ has a child labeled $(q',i)$, it then has no ancestor labeled $(q',i)$ whose parent is labeled $(q,i)$ (which avoids infinite trees in the presence of $\varepsilon$-loops). Ultimately, this ensures that a given $\varepsilon$-transition cannot be repeated before first consuming another input symbol, refer to Observation~\ref{obs:no-repeat-transition}.
The set of \emph{accepting runs} is the set of root to leaf paths in the tree, obtained from $T_A(w)$ and denoted as $\bar{T}_A(w)$, by deleting all nodes 
so that neither the node nor any of its descendants have a label from $F\times \{|w|\}$. We say that $A$ accepts $w$ if the accepting runs tree is non-empty. We denote by $\mathcal{L}(A)$ the set of all strings accepted by $A$.

For \emph{memoized runs of $A$ on $w\in \Sigma^*$}, we consider (possibly non-unique) trees $M_A(w)$ obtained from $T_A(w)$ by deleting the minimum number of nodes such that
there is at most one non-leaf labeled $(q,i)$, with $q\in \mathcal{M}$, where $\mathcal{M}$ are the memoized states of $A$. 
$\bar{M}_A(w)$ is obtained analogously from $\bar{T}_A(w)$.

We assume that all states are \emph{useful} in the mNFA considered here; that is, all states can be reached from the initial state, and a final state can be reached from all states.

\begin{observation}
    \label{obs:no-repeat-transition}%
    Condition $(iii)$ in the definition $T_A(w)$ may seem arbitrary and complex, preventing repeating epsilon transitions in a loop, without making use of memoization (which is of course also a possible way of avoiding these infinite loops). This, however, harmonizes with~\cite{van2021memoized} and with real-world matcher implementations (e.g.\ in Java), where $\mathcal{O}(|\delta|)$ space is used to prevent infinite loops, separate from the memoization discussed in Example~\ref{ex:java-memo} (later in this section). The simpler statement ``a node labeled $(q,i)$ cannot have an ancestor labeled $(q,i)$'' would be equivalent for all results here, but makes a difference for capturing semantics (i.e.\ parsing) in practice.
\end{observation}

\begin{observation} 
    \label{obs:worst-case}%
    $M_A(w)$ describes worst-case recursion trees, where we only reach a configuration $(q,i)$ for string position $i$ and a memoized state $q$ at most once. In practice, matching is done with a backtracking search where the transitions are \emph{prioritized}~\cite{inthewild}, that is, they are explored in a specific order and thus the various trees defined above, become ordered. 
    However, when considering all strings the worst matching time case can be realized in all but pathological cases. For example, matching $(\Sigma^*) \Or E$ will never use $E$ in practical matchers, so any worst-case derived from $E$ is not possible to exercise in such cases. As such we define matching time as if the tree is fully explored (see Definition~\ref{defn:matching-time}). Observe that $(a\Or a)^*$, for example, will match in linear time on $aa\cdots a$ (as every path succeeds), but exhibits its exponential matching time on $aa \cdots ab$. While pathological cases (such as $(.\Or .)^*$, with `$.$' matching any character) exist, they tend to correspond to expressions that are not meaningful, so we simplify our analysis by assuming full exploration of $M_A(w)$ (which leads to an over-approximation of matching time).
\end{observation}

The \emph{Thompson construction}~\cite{thompson1968} converts a regex $E$
into an NFA $\thom(E)$ inductively over the structure of $E$, using
$\varepsilon$-transitions to connect sub-automata. The resulting NFA has
$O(|E|)$ states and transitions, and every state corresponds to a symbol
or operator in $E$. A key structural property is that $\thom(E)$ is
\emph{reducible}~\cite{giammarresi2004thompson}: for any depth-first
search (DFS) of $\thom(E)$, the set of back edges is the same, which means that
each back edge corresponds to exactly one Kleene closure (or $+$)
operator in $E$. This property enables the application of linear-time
algorithms for minimum feedback vertex sets on Thompson
graphs~\cite{shamir1979linear}. Examples are given in Figure~\ref{fig:thompson-example}.
For any regex $A$, let $\entry(A)$ denote the node in $\thom(A)$ that receives an incoming $\eps$ edge in $\thom(A^*)$ and $\thom(A^+)$ (we assume this node exists in $\thom(A)$ and is the same for both closures, but this can be ensured in all Thompson constructions).
\begin{figure}[htb]
  \begin{center}
    \vskip1em
  \begin{tikzpicture}
    \begin{scope}[minimum width=4mm]
      \node[draw,circle] (q0) {};
      \node[draw,circle,right=of q0,xshift=1cm] (qF) {};
      \node[draw,circle,right=of qF] (q01) {};
      \node[draw,circle,double,right=of q01,xshift=1cm] (qF1) {};
    \end{scope}
  \node at ($(q0)!0.5!(qF)$) {$A$};
  \node at ($(q01)!0.5!(qF1)$) {$A'$};
  \draw ($(q0)-(5mm,0)$) edge[->] (q0);
  \draw (q0) edge[out=40,in=140,thick,gray,densely dashed] (qF);
  \draw (q0) edge[out=-40,in=220,thick,gray,densely dashed] (qF);
  \draw (q01) edge[out=40,in=140,thick,gray,densely dashed] (qF1);
  \draw (q01) edge[out=-40,in=220,thick,gray,densely dashed] (qF1);
  \draw (qF) edge[->,above] node {$\eps$} (q01);
  \begin{scope}[minimum width=4mm]
    \node[draw,circle,right=of qF1,xshift=.5cm] (q0) {};
    \node[draw,circle,right=of q0,xshift=1cm] (q1) {};
    \node[draw,circle,right=of q1,xshift=1cm] (q2) {};
    \node[draw,circle,double,right=of q2,xshift=1cm] (qF) {};
  \end{scope}
  \node[below=of q0,xshift=2em,yshift=1em] (elabl) {\scriptsize $\entry(A)$};
  \draw[densely dotted] (elabl) -- ($(elabl)!0.85!(q1)$);
  \node at ($(q0)!0.5!(qF)$) {$A$};
  \draw ($(q0)-(5mm,0)$) edge[->] (q0);
  \draw (q1) edge[out=40,in=140,thick,gray,densely dashed] (q2);
  \draw (q1) edge[out=-40,in=220,thick,gray,densely dashed] (q2);
  \draw[] (qF) edge[out=140,in=40,looseness=.5,<-,above] node {$\eps$} (q0);
  \draw[] (q0) edge[->,above] node {$\eps$} (q1);
  \draw[] (q2) edge[->,above] node {$\eps$} (qF);
  \draw[] (q2) edge[in=270,out=270,->,below] node {$\eps$} (q1);
  \end{tikzpicture}
  \end{center}
  \caption{An illustration of a few aspects of the Thompson construction: on the left the construction of $\thom(A \cdot A')$, given the subautomata $\thom(A)$ and $\thom(A')$. On the right the construction of $\thom(A^*)$ from the subautomaton $\thom(A)$. All operations except $a\in \Sigma$ and $\eps$ involve gluing subautomata together with epsilon transitions, and possibly new states (with $\thom(a)$ the unique minimal DFA for $\{a\}$). The exact set of states and epsilon transitions superficially differ between treatments, but for our purposes here it only matters that every subexpression corresponds to a subautomaton with a single distinguished entry state (the initial state for that automaton) and exit state (the unique accepting state for that automaton), and that the back edge added by a $*$- or $+$-closure goes to that entry state (i.e.\ in $\thom(A)$ above, an epsilon edge from the exit to the entry of $\thom(A)$ is added).
  }
  \label{fig:thompson-example}
\end{figure}
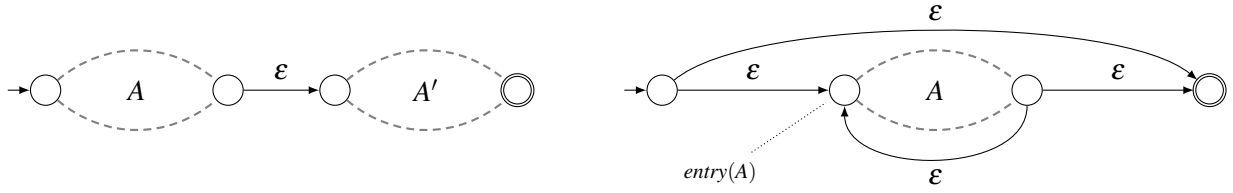

While the Thompson construction can vary somewhat depending on the source, the \emph{Glushkov construction}~\cite{glushkov1961abstract} varies little. For a full treatment, see, for example,~\cite{glushkov-tcs}. Here we recall the key properties for our purposes. For any regex $R$, the construction produces the epsilon-free NFA $\glush(R)$ obtained by first making each alphabet symbol in $R$ unique through \emph{linearization}. We simply sequentially number the symbols. For example, we turn $(abb)^* \Or (bc)$ into $(a_1b_2b_3)^* \Or (b_4c_5)$. We create one distinct initial state $q_0$ and one state for each alphabet symbol. In our example, $Q=\{q_0,a_1,b_2,b_3,b_4,c_5\}$. Each non-initial state then represents the language of suffixes that can be matched if we `start' matching just to the right of that alphabet symbol occurrence in the expression. In the example, state $a_1$ accepts the strings in $b_2b_3(a_1b_2b_3)^*$, while state $b_4$ accepts only $\{c_5\}$. Both $\delta$ and $F$ are uniquely determined by this. $\delta$ describes simple reachability in the expression, and $F$ indicates whether the end of the expression is reachable (without reading any further input symbols). Dropping the indices from the input alphabet produces $\glush(R)$. In this example, we would have $\delta(q_0,a)=\{a_1\}$, $\delta(q_0,b)=\{b_4\}$, $\delta(a_1,b)=\{b_2\}$, $\delta(b_2,b)=\{b_3\}$, $\delta(b_3,a)=\{a_1\}$, $\delta(b_4,c)=\{c_5\}$, all other entries $\emptyset$, and $F=\{q_0,b_3,c_5\}$.

The key part we require here is that the unique initial state never receives an incoming transition (as all transitions read some symbol, and that symbol must be one of the occurrences in the expression). However, obtaining $\glush(A \cdot A')$ from $\glush(A)$ and $\glush(A')$ involves an arbitrary number of transitions (as any final state in $\glush(A)$ would have transitions to all symbols reachable from the initial state in $\glush(A')$). Analogously, $\glush(A^*)$ will have each final state in $\glush(A)$ gain the transitions of the initial state in $\glush(A)$. However, for example, $\glush(A \Or A')$ is isomorphic to what is obtained by merging the initial states of $\glush(A)$ and $\glush(A')$.

Given the frequent use of abbreviations in this paper, the key acronyms are categorized and briefly described next. IDA (and EDA) describe the ambiguity of automata (Definitions~\ref{def:ida} and~\ref{def:ambi}), which causes redundant matching attempts and, as a result, an excessive number of matching steps. All of \None, \All, \IN, \CN, \I{}, \MFN, and \IR+ are memoization schemes from Definitions~\ref{def:memoschemes}, \ref{def:mfn}, and~\ref{def:IAR+}. Each determines a subset of the states of an automaton, the memoization of which removes IDA. Finally, MFVS, the minimum feedback vertex set of a graph, is defined at the end of this section and underpins the memoization scheme \MFN{}.

Let us first define the matching time of an NFA. As this will only be used for worst-case analysis, we keep it simple and assume the entire tree of memoized runs is visited. This may not be the case for every string, but for expressions where all subexpressions are useful, this is quite realistic; recall Observation~\ref{obs:worst-case}.
\begin{definition}[Matching time]\label{defn:matching-time}
The \emph{matching time} of the NFA $A$ on input $w$ is the worst-case number of nodes in $M_A(w)$ (recall, $M_A(w)$ is not uniquely defined).
$A$ has linear matching time if there exists a constant $c$ such that, for every input $w$, the matching time is at most $c\cdot|w|$.
\end{definition}

For a tree $T$, let $|T|_{\textit{leaf}}$ denote the number of leaf nodes in $T$.

\begin{definition}[Ambiguity and IDA~\cite{van2021memoized}]
  \label{def:ambi}%
  The \emph{ambiguity} of an mNFA $A$, denoted $a_A(n)$, is defined as $\max \{ |\bar{M}_{A}(w)|_{\textit{leaf}} \mid w\in \Sigma^*, |w|\le n\}$.
  If $a_{A}(n)$ is unbounded, $A$ has
  \emph{infinite degree of ambiguity} (IDA). 
\end{definition}

\begin{theorem}[IDA characterization~\cite{van2021memoized}]
  \label{def:ida}
  An mNFA $A$ has IDA if and only if there exist two distinct states $p$
  and $q$, and a string $v\in\Sigma^+$, such that there are paths on $v$
  from $p$ to $p$, from $p$ to $q$, and from $q$ to $q$, with the cycle
  on $q$ containing no memoized state.
\end{theorem}

When no states are memoized, \changeB{Theorem}~\ref{def:ida} reduces to the classical IDA
criterion for NFAs of Weber and Seidl~\cite{weber1991degree}. 
In \cite{van2021memoized}, it was shown that IDA for a $\varepsilon$-loop-free mNFAs is decidable in
cubic time and that finding a
\emph{minimum} set $M$ (that is, a set of the least cardinality) to eliminate IDA is NP-complete for general NFAs.

\begin{theorem}[Linear matching time iff no IDA~\cite{van2021memoized}]
An mNFA A has linear matching time on all inputs if and only if $A$ does not have IDA. %
\end{theorem}

\begin{definition}[\IN{}, \CN{}, \I{}~\cite{van2021memoized}]
  \label{def:memoschemes}%
  For an mNFA $A=(\mathcal{M},Q,\Sigma,\delta,q_0,F)$, we consider the following sets $\mathcal{M}$ to be memoized:
  \begin{enumerate}
    \item $\textbf{\None{}}$ is the trivial case of $\mathcal{M}=\emptyset$;
    \item $\textbf{\IN{}}$ is the set of all states with in-degree $\geq 2$; 
    \item \textbf{\CN{}} (Closure Nodes) is obtained by considering all possible depth-first search orderings on $A$, and memoizing all states which are reached by a backwards edge (that is, on a path from $q$ we encounter an edge to $q$); 
    \item \textbf{\I{}} (Infinite Ambiguity Removal) is a set of states obtained by iterating over states (in \emph{some} order) in \CN{}, checking, for each, whether its removal from the set of states being memoized reintroduces IDA; and 
    \item 
$\textbf{\All{}}$ is the trivial case of taking $\mathcal{M}=Q$. 
  \end{enumerate}
\end{definition}

\begin{remark}
\None{} is used by many backtracking parsers in practice, such as the standard regular expression library in Python. Observe that $\mathrm{\CN{}}\subseteq\mathrm{\IN{}}$.
Also, \All{} is a safe but expensive memoization scheme and is not considered in the remainder of this paper.
\end{remark}

\CN{} memoization ensures that every cycle in $A$ contains at least one
memoized state, which by \changeB{Theorem}~\ref{def:ida} eliminates IDA. By definition, \I{} yields a
minimal memoization (i.e.\ we cannot remove any of the memoized states without reintroducing IDA) that might not be minimum (i.e.\ there might be sets of states of smaller cardinality that also remove IDA). We show in Section~\ref{sec:iar} that \I{} starting from \CN{} (on Glushkov) does not \emph{necessarily} produce a set of the smallest cardinality that will remove IDA. Note that~\I{} operates in some unspecified order as defined in~\cite{van2021memoized}, and different orders may yield different results. Picking a \emph{good} order is difficult (unless $\textrm{P}=\textrm{NP}$). If \I{} processes states in an order that exhausts the complement of some minimum set first, it produces that minimum set.
Since \I{} iterates over states in \CN{} and checks, for each one, whether its removal reintroduces IDA, \I{} runs in polynomial time. Each IDA check is $O(|Q|^3)$, and the number of iterations is at most $|\CN{}| \le |Q|$; thus, a naive implementation will run in time $O(|Q|^4)$.

\begin{example}
    \label{ex:java-memo}%
    Like most practical implementations, the regex matcher in Java uses backtracking search, and prior to version~9 had no protection against catastrophic backtracking scenarios. An inspection of the source code of the regex matcher in Java~9 revealed that memoization was implemented by picking a subset of \CN{} by considering only ``top-level'' closure nodes with certain characteristics. As a result, regexes such as
$((a\Or a)^*)^*$ and $((a^+)^+)^+$ still trigger catastrophic
backtracking in Java~9 and later because inner
closure nodes are excluded from memoization. This illustrates that even in industrial
practice, the question of \emph{which} states to memoize remains
important.
\end{example}

A \emph{feedback vertex set} %
of a directed graph $G=(V,E)$ is a
subset $S\subseteq V$ such that every directed cycle in $G$ contains at
least one vertex from $S$, or equivalently, $G[V\backslash S]$ (i.e.\ the subgraph of $G$ with nodes $S$ and the corresponding edges removed) is
acyclic. A \emph{minimum feedback vertex set} (MFVS) is a feedback vertex set of
minimum cardinality. Computing an MFVS is NP-complete for general
directed graphs~\cite{karp1972}, but it is solvable in linear time for
\emph{reducible} graphs~\cite{shamir1979linear}. A directed rooted graph
is \emph{reducible} if the set of back edges is the same for every DFS
ordering.

\section{MFN Memoization and the PTA Algorithms}
\label{sec:mfn}

We show that when using Algorithms~\ref{alg:thompson-pta} and~\ref{alg:glushkov-pta} (discussed in this section) to compute an MFVS, it is contained in \CN{}. In the remainder of this paper, when we refer to \MFN{}, we assume that it was computed using these algorithms. It should be pointed out that Algorithm~\ref{alg:thompson-pta} does not compute an MFVS in the presence of $\varepsilon$-loops. Those cases are handled separately in practice, but they could, of course, also be handled by memoization --- see Observation~\ref{obs:no-repeat-transition}.

\begin{definition}[MFN Memoization]
  \label{def:mfn}
  For Thompson or Glushkov NFAs, the \emph{Minimum Feedback Node} (\MFN{}) memoization scheme is obtained by using 
  Algorithms~\ref{alg:thompson-pta} or~\ref{alg:glushkov-pta}, respectively. 
\end{definition}

Since every (non-$\varepsilon$) cycle in $A$ contains at least one state
in the \MFN{} set, it follows from \changeB{Theorem}~\ref{def:ida} that \MFN{} eliminates IDA. Thus
\MFN{} provides correctness guarantees equivalent to \CN{} while often using
fewer (memoized) states. 
We note that an MFVS need not be unique; the PTA algorithms presented
below compute a specific, canonical MFVS (when ignoring $\eps$-cycles).
Since Thompson graphs are reducible~\cite{giammarresi2004thompson},
Shamir's algorithm~\cite{shamir1979linear} computes an MFVS in linear
time when given the state machine as input. The Thompson PTA (Algorithm~\ref{alg:thompson-pta}) achieves the same result from the \emph{regex parse tree}, which is convenient for integration into regex compilers
that build the parse tree before building the NFA. 
The algorithm tracks, for each subexpression $R$:
(i) $R_{\mathrm{req}}$, the MFVS computed so far for the subgraph of
    $R$;
(ii) $\mathit{isCovered}(R)$, whether every path through $R$'s
    subgraph, reading at least one symbol from $\Sigma$, already contains a state in $R_{\mathrm{req}}$; and
(iii) $\mathit{isNullable}(R)$, whether $\varepsilon\in \mathcal{L}(R)$.

\begin{algorithm}[t]
\caption{Thompson PTA}
\label{alg:thompson-pta}
\begin{algorithmic}[1]
\Require{Parse tree of regex $R$}
\Ensure{$(\mathit{isCovered}(R),\; R_{\mathrm{req}},\; \mathit{isNullable}(R))$}
\If {$R = \varepsilon$}
  \Return $(\mathit{True},\; \emptyset,\; \mathit{True})$
\ElsIf{$R = a \in \Sigma$}
  \Return $(\mathit{False},\; \emptyset,\; \mathit{False})$
\ElsIf{$R = A^?$}
 \State $c,r,n \gets \Call{ThompsonPTA}{A}$
 \State \Return $(c,\; r,\; \mathit{True})$
\ElsIf{$R = A^*$ \textbf{or} $R = A^+$}
  \State $c,r,n \gets \Call{ThompsonPTA}{A}$
  \If{$\lnot c$}
    \Return \label{alg:line:r-star} $(\mathit{True},\; 
    r \cup \{\entry(R)\},\; n \lor (R = A^*))$ 
  \Else \, \Return $(c,\; r,\; n \lor (R = A^*))$
  \EndIf
\ElsIf{$R = A \Or A'$}
  \State $c,r,n \gets \Call{ThompsonPTA}{A}$;\,
        $c',r',n' \gets \Call{ThompsonPTA}{A'}$
\State \Return $(c \land c',\; r \cup r',\; n \lor n')$
\ElsIf{$R = A \cdot A'$ \label{alg:line:a-cdot-ap}}
  \State $c,r,n \gets \Call{ThompsonPTA}{A}$;\,
         $c',r',n' \gets \Call{ThompsonPTA}{A'}$
  \If{$n$ and $n'$} \Return $(c \land c',\; r \cup r',\; \mathit{True})$
  \ElsIf{$n'$}      \Return $(c,\; r \cup r',\; \mathit{False})$
  \ElsIf{$n$}       \Return $(c',\; r \cup r',\; \mathit{False})$
  \Else\,\Return $(c \lor c',\; r \cup r',\; \mathit{False})$
  \EndIf
\EndIf
\end{algorithmic}
\end{algorithm}

The key insight is the \emph{domination property} of Thompson graphs:
$\entry(A)$, the state corresponding to a closure operator $R = A^*$ (or $A^+$),
dominates every state inside $A$'s subgraph, in the sense that any path
entering $A$'s subgraph must first pass through $R$'s closure state.
Hence, if all paths through $A$ already contain a memoized state
($\mathit{isCovered} = \mathit{True}$), the closure state should not be
added to $R_{\mathrm{req}}$; otherwise, it must.

Let us put all of this together to demonstrate correctness. Ignoring epsilon cycles has the obvious meaning; where the MFVS definition states ``every directed cycle,'' we restrict our attention to cycles that contain at least one transition labeled by a symbol from $\Sigma$. 
\begin{theorem}
  \label{thm:thompson-mfvs}%
  For every regex $R$, the set $R_{\mathrm{req}}$ computed by Algorithm~\ref{alg:thompson-pta} is an MFVS of $\thom(R)$ when ignoring epsilon cycles. Furthermore, $R_{\mathrm{req}} \subseteq \mathrm{\CN{}}(\thom(R))$.
\end{theorem}
\begin{proof}
  We proceed by structural induction on the parse tree of $R$. For each subexpression $E$, the algorithm tracks triples $(c,E_{\mathrm{req}},n)$ where $E_{\mathrm{req}}$ is a MFVS for $E$. The covered status (\textit{isCovered}) $c$ is true iff all non-$\eps$ \emph{paths} in $E$ contain at least one state from $E_{\mathrm{req}}$, and the nullability (\textit{isNullable}) $n$ is true iff $\eps\in \mathcal{L}(E)$. The algorithm tracks all three together to perform the computation in a single linear pass, but we can consider the parts separately. Nullability, the $n$, is as usual and can easily be separately verified by the reader (note the case on line~\ref{alg:line:r-star}, which amounts to $E^*$ always being nullable but $E^+$ being nullable iff $E$ is nullable).
  
  Next, $E_{\mathrm{req}}$ is an MFVS for $E$. For the base cases, $E\in \Sigma \cup \{\varepsilon\}$ this is trivially true, as $\thom(E)$ then has no loops and $E_{\mathrm{req}}=\emptyset$. The cases for $E=A^?$, $E=A\Or A'$, and $E=A \cdot A'$ similarly never create a loop under the Thompson construction (see Figure~\ref{fig:thompson-example}), making the union of the vertex sets covering the non-$\eps$ loops of the parts cover the non-$\eps$ loops of the new NFA. The cases that create loops are $E=A^*$ and $E=A^+$, and these will clearly create a non-$\eps$ loop that needs to be covered iff $A$ has a non-$\eps$ path that does not contain an already covered state. This is tracked by the covered part of the triple, $c$, so if $c$ is false, we add the entry state of $\thom(A)$ to $E_{\mathrm{req}}$, as the entry state dominates the subgraph; all loops in $E^*$ must use that state. To see that this vertex feedback set is minimal, simply observe that each subautomaton that gets a memoized state (i.e.\ it corresponds to a *- or +-closure) cannot ``lose'' that memoized state, as it then has an uncovered cycle that is necessarily both reachable and useful. As such, we cannot avoid adding the memoized states we add, and the way the construction continues cannot later make it unnecessary.
  
   Finally, the correctness of the covered status $c$ does not depend on $E_{\mathrm{req}}$, avoiding circularity. As base cases, $E=\alpha \in \Sigma$ is not covered, but $E=\eps$ is, since $E=\alpha$ has a path labeled $\alpha$ that contains no states in $E_{\mathrm{req}}=\emptyset$, whereas $\eps$ has no non-$\eps$ path. Proceeding inductively, we have a straightforward case analysis; for example, the $E=A\cdot A'$ case on line~\ref{alg:line:a-cdot-ap}, if both $A$ and $A'$ are covered, then so is $E$, and conversely, if neither is covered, neither is $E$. If $A$ is covered but nullable, and $A'$ is not covered, the path for $\eps$ in $A$ combined with any non-covered non-$\eps$ path in $A'$ creates a non-covered non-$\eps$ path in $E$, making it not covered (the same argument can be repeated with $A$ and $A'$ reversed).

  We have $R_{\mathrm{req}} \subseteq \mathrm{\CN{}}(\thom(R))$ as the only states added to any $E_{\mathrm{req}}$ are $\entry(A)$ for some $E=A^*$ or $E=A^+$, which is precisely when a back edge is added to the state $\entry(A)$. Refer to Figure~\ref{fig:thompson-example}.
\end{proof}

Since Glushkov NFAs are not always reducible, Shamir's algorithm does
not apply directly. Instead, we (again) use the parse tree structure of the
regex to compute the MFVS of $\glush(R)$.
For each subexpression $R$, the algorithm computes:
(i) $R_{\mathrm{req}}$, the MFVS of $\glush(R)$;
(ii) $R_{\mathrm{op}}$, the set of states that would need to be added to $R_{\mathrm{req}}$ if a closure operator were applied to $R$, i.e.\ $R^*_{\mathrm{req}} = R_{\mathrm{req}} \cup R_{\mathrm{op}}$, and;
(iii) $\mathit{isNullable}(R)$.

\begin{algorithm}[t]
\caption{Glushkov PTA}
\label{alg:glushkov-pta}
\begin{algorithmic}[1]
\Require{Parse tree of regex $R$}
\Ensure{$(R_{\mathrm{op}},\; R_{\mathrm{req}},\; \mathit{isNullable}(R))$}
\If{$R = \varepsilon$}
\Return $(\emptyset,\; \emptyset,\; \mathit{True})$
\ElsIf{$R = a$}
  \Return $(\{a_1\},\; \emptyset,\; \mathit{False})$ \Comment{Where $a_1$ is the state reached on reading $a$}
\ElsIf{$R = A^?$}
 \State $o,r,n \gets \Call{GlushkovPTA}{A}$
 \State \Return $(o,\; r,\; \mathit{True})$
\ElsIf{$R = A^*$ \textbf{or} $R = A^+$}
  \State $o,r,n \gets \Call{GlushkovPTA}{A}$
  \State \label{line:gl:stars} \Return $(\emptyset,\; o \cup r,\; n \lor (R = A^*))$
\ElsIf{$R = A \Or A'$}
  \State $o,r,n \gets \Call{GlushkovPTA}{A}$
  \State $o',r',n' \gets \Call{GlushkovPTA}{A'}$
  \State \Return $(o \cup o',\; r \cup r',\; n \lor n')$
\ElsIf{$R = A \cdot A'$}
  \State $o,r,n \gets \Call{GlushkovPTA}{A}$;\,
         $o',r',n' \gets \Call{GlushkovPTA}{A'}$
  \If{$n$ and $n'$}   \Return $(o \cup o',\; r \cup r',\; \mathit{True})$
  \ElsIf{$n'$}        \Return $(o,\; r \cup r',\; \mathit{False})$
  \ElsIf{$n$}         \Return $(o',\; r \cup r',\; \mathit{False})$
  \Else\If{$|o| \leq |o'|$} \Return $(o,\; r \cup r',\; \mathit{False})$
    \Else\,                \Return $(o',\; r \cup r',\; \mathit{False})$
    \EndIf
  \EndIf
\EndIf
\end{algorithmic}
\end{algorithm}

The $R_{\mathrm{op}}$ set tracks the smallest set of states that
cover all paths through $R$'s subgraph, enabling the algorithm to make
an informed choice at the concatenation operator when neither operand is
nullable. In line 2, $\{a_1\}$ denotes the set consisting of the final state of $\glush(a)$.

\begin{theorem}
  \label{thm:glushkov-mfvs}%
  For every regex $R$, the set $R_{\mathrm{req}}$ computed by Algorithm~\ref{alg:glushkov-pta} is an MFVS of $\glush(R)$.
  Furthermore, $R_{\mathrm{req}} \subseteq \mathrm{\CN{}}(\glush(R))$.
\end{theorem}
\begin{proof}
  We proceed by structural induction on the depth of the parse tree of $R$. For each subexpression $E$ the algorithm computes the tuple $(E_{\mathrm{op}}, E_{\mathrm{req}}, n)$, it is again easy for the reader to verify that the nullability of the expression, $n$, is computed correctly. It remains to show that $E_{\mathrm{req}}$ is an MFVS 
  for $\glush(E)$, and that $E_{\mathrm{op}}$ similarly is a minimum cardinality set containing at least one state from every non-epsilon path through $\glush(E)$.

  We begin with $E_{\mathrm{req}}$, assume $E_{\mathrm{op}}$ is correct and minimal. For the base cases $E=\eps$ is a single accepting state and $E=a\in \Sigma$ is the unique minimal DFA, neither of which contains cycles, so $E_{\mathrm{req}}=\emptyset$.

  The NFA $\glush(A \Or A')$ can be obtained as the union of $\glush(A)$ and $\glush(A')$, merging their initial states. Since the initial state has no incoming transitions this creates no new cycles or paths, and the union of the states covering the cycles/paths of the subautomata will cover the cycles and paths of the resulting automaton. $E=A^?$ can be obtained by marking the initial state of $\glush(A)$ final, causing no change to cycles or paths in $\glush(A)$. For $E=A^*$ and $E=A^+$ we rely on assuming $A_{\mathrm{op}}$ correct: under the closure every path becomes a cycle, and as $A_{\mathrm{req}}$ covers all 
  cycles and $A_{\mathrm{op}}$ covers all 
  non-epsilon paths 
  not already covered by $A_{\mathrm{req}}$ we have $A_{\mathrm{op}} \cup A_{\mathrm{req}}$ covering all non-epsilon cycles in $E$, it remains to argue that if $A_{\mathrm{op}}$ and $A_{\mathrm{req}}$ are both minimal then so is $E_{\mathrm{req}}$. This is the case since the paths giving rise to states in $A_{\mathrm{op}}$ and the cycles giving rise to states in $A_{\mathrm{req}}$ are necessarily independent: states are only freshly added to $A_{\mathrm{req}}$ on line~\ref{line:gl:stars}, where $A_{\mathrm{op}}$ is simultaneously emptied. New uncovered paths may then be added, by union or concatenation, but neither case will intertwine the paths with existing cycles (in the case of concatenation the path may ``reach'' the cycle, but it will either fully follow the cycle, or fully pass it by). Finally, for $E=A \cdot A'$ there are no transitions from $A'$ to $A$, and the transitions added go from $A$ to $A'$, so no cycles can be created, so $E_{\mathrm{req}}$ is simply the covers for the parts, $A_{\mathrm{req}} \cup A'_{\mathrm{req}}$.

  It remains to show $E_{\mathrm{op}}$ correct and minimal, observe that this does not rely on the correctness of $E_{\mathrm{req}}$, causing no circularity. For $E=\eps$ no path exists, but for $E=a\in \Sigma$ a path labeled $a$ exists, and the state $a_1$ is added to $E_{\mathrm{op}}$. $E=A^?$ and $E=A\Or A'$ proceed similarly to $E_{\mathrm{req}}$. For $E=A^*$ and $E=A^+$ observe that $E$ will have no path that is not also on a cycle, so $E_{\mathrm{op}}=\emptyset$ as all paths are covered by $E_{\mathrm{req}}$. The one complex case is $E=A \cdot A'$, where paths get concatenated, and the nullability of $A$ and $A'$ comes into play. If $A$ is nullable then the paths of $A'$ all need to be covered in $E$ as well, and vice versa. However, if neither $A$ nor $A'$ is nullable all non-epsilon paths will consist of non-epsilon paths through $\glush(A)$ and $\glush(A')$ glued together, and the path can be covered in either, as such we pick the smaller of $A_{\mathrm{op}}$ and $A'_{\mathrm{op}}$. Observe that the full cartesian product of paths exists (that is, every choice of one path from $A$ and one from $A'$ together form a path through $A\cdot A'$), so we cannot do better by covering some paths in $\glush(A)$ and some in $\glush(A')$, if there is a single uncovered path in each that creates an uncovered path in $\glush(E)$.

  Finally, we can show that for each subexpression $E$, $E_{\mathrm{op}} \subseteq \mathrm{\CN{}}(\glush(E^*))$ and, since $R_{\mathrm{req}}$ is entirely constructed from unions of these $E_{\mathrm{op}}$ sets, we have that $R_{\mathrm{req}} \subseteq \CN{}$. The base cases $E=\eps$ and $E=a\in \Sigma$ are trivial. 
  The cases of $E=A^?$, $E=A^*$, and $E=A^+$ follow directly from the inductive assumption. 
  In the case of $E=A\Or A'$ it is clear that $\mathrm{\CN{}}(\glush(A)) \subseteq \mathrm{\CN{}}(\glush(E))$ and $\mathrm{\CN{}}(\glush(A')) \subseteq \mathrm{\CN{}}(\glush(E))$ and thus we have the desired result. 
  The last case is that of $E=A \cdot A'$.
  When either $A$ or $A'$ is nullable then $E_{\mathrm{op}}$ is selected such that a DFS ordering of $\glush(E^*)$ exists where the states in $E_{\mathrm{op}}$ are visited first and then have back-edges incident either via the back-edges created by the closure operator or via the nullable paths through the corresponding subgraph.
  In the case where neither $A$ nor $A'$ is nullable then the smallest cardinality set is selected as $E_{\mathrm{op}}$. When $|A_{\mathrm{op}}| \leq |A'_{\mathrm{op}}|$ we select $A_{\mathrm{op}}$ which are states that can be visited first in a DFS ordering and thus have the desired property following the same reasoning as above.
  Then the last case would be $|A_{\mathrm{op}}| > |A'_{\mathrm{op}}|$. If $A'_{\mathrm{op}} = \emptyset$ then we have the result trivially. Thus if $A'_{\mathrm{op}} \neq \emptyset$, then $|A_{\mathrm{op}}| \geq 2$ which means that $A$ has at least two different paths. As a result, applying a closure operator and doing a DFS will allow for two different ways to reach the states of $A'_{\mathrm{op}}$. In the DFS ordering we can follow the first path, reach the states of $A'_{\mathrm{op}}$, then follow the new back-edges back to states of $A_{\mathrm{op}}$, follow the second path, and again reach the states of $A'_{\mathrm{op}}$. This then leaves us with the desired result for all cases.

\end{proof}

\begin{corollary}
  Both PTAs run in $O(|E|)$ time and space, where $|E|$ is the size of
  the regex.
\end{corollary}

This matches the time complexity of determining the memoized states for \CN{} and \IN{}, making \MFN{} a
drop-in replacement with at-most-equal and often  strictly smaller memory usage.

\begin{example}[Thompson vs.\ Glushkov under \MFN{}]
  \label{ex:a-plus-a-plus}
  For $(a^+)^+$, the Thompson NFA has two closure nodes (for the inner
  and outer positive closures). 
  Once the inner
  $+$-state is memoized, the outer $+$-cycle is also broken. Thus
  $|\MFN{}(\thom((a^+)^+))| = 1$. The Glushkov NFA for $(a^+)^+$
  collapses, apart from the initial state, to a single state with a self-loop, and thus we also have  $|\MFN{}(\glush((a^+)^+))| = 1$. In general, when using the Thompson construction,  we memoize either the same number or fewer states compared to when using the Glushkov construction (see Section~\ref{sec:thompson-vs-glushkov}).
\end{example}

\section{IAR and Minimal vs Minimum Memoization}
\label{sec:iar}

The \I{} scheme~\cite{van2021memoized} takes \CN{} as a starting point and
iteratively removes states whose memoization is not required to prevent
IDA. Formally, a state $q \in \mathrm{\CN{}}$ is \emph{removable} if the
mNFA obtained by un-memoizing $q$ still satisfies the condition of
\changeB{Theorem}~\ref{def:ida} being false. Although \I{} is guaranteed to
eliminate IDA and to be a subset of \CN{}, it is \emph{not} guaranteed to
produce a set of the smallest cardinality that removes IDA after memoization. We demonstrate this
with the following example.

\begin{example}
  \label{ex:iar-not-minimal}

Consider $r_1 = ((aa^? \Or aa^?)aa^?)^*$ with the Glushkov construction. 
 Linearization of $r_1$ gives the states $a_1,\ldots, a_6$ corresponding
  to the six occurrences of $a$. \CN{} memoizes $\{a_1, a_3, a_5\}$
  (the states reachable via back edges in some DFS). Note, in some DFS ordering, $a_5$ does appear as a back-edge target. Starting from
  \CN{} and considering the order $\{a_5, a_3, a_1\}$ will only remove $a_5$ and thus memoize $\{a_1, a_3\}$. However, \MFN{} only memoizes $\{a_5\}$, the unique
  MFVS of $\glush(r_1)$. Since memoizing $a_5$ alone breaks every cycle in
  $G(r_1)$, IDA is eliminated with a single memoized state. Thus, \I{}
 produces a strictly larger set than \MFN{} in this example. As mentioned earlier, \I{} is sensitive to the order of state consideration. Had $\{a_5\}$ not been considered first, \I{} would result in the same set as \MFN{}.

\end{example}

\begin{example}
  \label{ex:non-ida}
  
  For $r = ((ab^? \Or cd^?)ef^?)^*$, \MFN{} memoizes one state (the state for
  $e$ for Glushkov and the state corresponding to the closure operator for Thompson), but no memoization is
  required at all. This is the case since the IDA condition requires two states $p\neq q$
  with a common loop string $v$, which cannot occur here as all paths use
  distinct symbols.
\end{example}

We distinguish between two notions of optimality:
(i) A memoization set $M$ is \emph{minimal} if no proper subset of
    $M$ also eliminates IDA; and
(ii) a memoization set $M$ is \emph{minimum} if it has the least cardinality among all sets eliminating IDA.

Every minimum set is minimal, but the converse may fail. \MFN{} produces a
minimum feedback vertex set (in the absence of $\varepsilon$-loops) for the \emph{graph-structural} criterion (no cycle
uncovered), but the true minimum memoization may be smaller if some
cycles do not contribute to IDA. 
For example, $a^*$ has a single cycle, but no
memoization is required since there is no IDA: the regex matches any
string $a^n$ with exactly one accepting run. Thus, \MFN{} is not always
minimal.
This motivates the following two equivalent conjectures.

\begin{conjecture}
  \label{conj:mfn-contains-min}
  For every regex $R$ and every minimum memoization set $H$ for
  $\thom(R)$ (respectively $\glush(R)$), there exists a minimum memoization set
  $H'$ with $|H'|=|H|$ and $H'\subseteq \mathrm{\MFN{}}(\thom(R))$
  (respectively $H'\subseteq \mathrm{\MFN{}}(\glush(R))$).
\end{conjecture}

Conjecture~\ref{conj:mfn-contains-min} is supported by an exchange
argument: given any minimum memoization set $H$ containing a state $x
\notin \mathrm{\MFN{}}$, one can attempt to swap $x$ for the \MFN{} state that
covers the same cycles as $x$ (which exists due to the minimality of the
MFVS). 

\begin{definition}\label{def:IAR+}
The set of states to memoize in \IR+{} is computed as follows.
(i) Compute $\mathrm{\MFN{}}(A)$; and
(ii) place an order on the states from (i), and in order, remove the state under consideration if, by removing them, we will not re-introduce IDA.
\end{definition}

\begin{conjecture}
  \label{conj:mfn-refine}
  For some order on the states, \IR+{} produces a minimum memoization set.
\end{conjecture}

\begin{observation}
We note that Conjecture~\ref{conj:mfn-contains-min} implies Conjecture~\ref{conj:mfn-refine}. A required order is obtained by ordering the states in $\MFN{}$ by having last a sequence of states, from a set of minimum cardinality that removes IDA.
But it also follows directly from Definition~\ref{def:IAR+} that Conjecture~\ref{conj:mfn-refine} implies Conjecture~\ref{conj:mfn-contains-min}.
\end{observation}

\begin{example}
  \label{ex:alternation-closure}
  In a more specific variant of Example~\ref{ex:non-ida}, for $r = ((a\Or b \Or c \Or d) e)^*$ with Glushkov, \CN{} memoizes the four
  alternation states $\{a,b,c,d\}$ and $e$, while \MFN{} memoizes
  only $\{e\}$, since $e$ alone breaks every cycle. This example illustrates that \MFN{} is effective when having a closure over a large
  alternation concatenated with a non-nullable subexpression.
\end{example}

\begin{example}
  \label{ex:nested-closures}
  For $r = (((a^+)^+)^+)$ with Thompson, \CN{} memoizes all three closure
  nodes, while \MFN{} memoizes only the innermost closure node, whose
  removal disconnects all back edges (by the domination property). 
\end{example}

\section{Comparing Thompson and Glushkov under MFN}
\label{sec:thompson-vs-glushkov}

\begin{theorem}
  \label{thm:thompson-leq-glushkov}
  For any regex $R$,
  $|\mathrm{\MFN{}}(\thom(R))| \leq |\mathrm{\MFN{}}(\glush(R))|$.
\end{theorem}

\begin{proof}[Proof sketch]
  For Thompson, each closure operator $A^*$ or $A^+$ contributes exactly
  one new back-edge
  in $\thom(R)$. This back-edge is incident to a closure state which
  is selected to cover the cycles in $\thom(R)$ unless the subgraph of $A$ is already covered. The Thompson PTA therefore selects at most one state per closure operator, equaling the number of uncovered closures.

  For Glushkov, the closure operators $A^*$ and $A^+$ do not contribute a dedicated state;
  instead, the back edges connect states within the subgraph of $A$, reachable from the initial state of $A$. The Glushkov PTA must therefore add at least one state
  \emph{per uncovered closure}, but may add more if the subgraph has
  multiple ``entry points'' 
  that must all
  be covered. 
\end{proof}

\begin{corollary}
  The same bound holds for \CN{}: $|\mathrm{\CN{}}(\thom(R))| \leq
  |\mathrm{\CN{}}(\glush(R))|$ for any regex $R$.
\end{corollary}

\begin{example}
  For $(a \Or a)^*$: $\thom((a \Or a)^*)$ has one closure state (the $*$-node),
  giving $|\mathrm{\MFN{}}| = |\mathrm{\CN{}}| = 1$. For $\glush((a \Or a)^*)$, both
  states $a_1$ and $a_2$ are reachable via the back edges, giving
  $|\mathrm{\CN{}}| = 2$. By Algorithm~\ref{alg:glushkov-pta}, the Glushkov
  PTA would select $\{a_1\}$,
  giving $|\mathrm{\MFN{}}| = 1$ in this case, matching Thompson.
\end{example}

\begin{example}
  For $((aa^? \Or aa^?)aa^?)^*$: As in Example~\ref{ex:iar-not-minimal}, the
  Glushkov \MFN{} is $\{a_5\}$, while \CN{} gives $\{a_1,a_3,a_5\}$. For the
  Thompson NFA, \MFN{} and \CN{} both give a single closure state. This example
 illustrates that, for Glushkov, the difference between \MFN{} and \CN{} can be
  significant even when the Thompson NFA does not make such a distinction.
\end{example}

\section{Experimental Results}
\label{sec:experiments}

Theoretical analysis shows that \MFN{} memoization, computed by the PTA algorithms, results in a smaller set than competing memoization schemes (aside from \I{}) and determines these nodes in equal or improved time complexity. Our experiments follow a procedure similar to that of Roodt et al.~\cite{roodt2024benchmarking}. Here we consider the Polyglot corpus of regular expressions collected for~\cite{davis2019aren} that contains 537\,806 regular expressions from online repositories. The PTA algorithms have been implemented in BRU (Brendan's Regex Utility)\footnote{https://github.com/bkmwatling/srvm}, which will be the matcher used to conduct these experiments. Due to some syntax limitations in BRU, experiments were conducted on the 440\,079 successfully parsed expressions. For each regex, 20 strings were generated using Xeger\footnote{https://pypi.org/project/xeger/}: 10 positive (accepting) and 10 negative (rejecting) by mutation. Full and partial matching was performed for all inputs. Table~\ref{tab:strings} indicates the mean, median, and standard deviation of the size of the strings sampled for use during the experiments.

Figures~\ref{fig:memosize} and \ref{fig:state} illustrate the comparison in the number of memoized entries (or matching steps) used. For each cell $(x, y)$, $x$ is the number of memoized entries (or matching steps) used by \MFN{} and $y$ is the number of memoized entries (or matching steps) used by either \IN{} or \CN{} memoization. Cell luminosity is proportional to frequency, with lighter shades indicating a higher density of occurrences for a given pair. We use IQR (interquartile range) outlier removal for these graphs to focus attention on the most common scenarios without skewing the figures.

The main motivation for \MFN{} is the reduction of memoized nodes, which reduces the memory usage of the matcher. In every case, \MFN{} used the same number of or fewer memoization entries compared to \IN{} and \CN{}. This is illustrated in Figure~\ref{fig:memosize}, where each experiment lies above the diagonal line. 
Table~\ref{tab:size} indicates the mean, median, and standard deviation of the number of memoization entries for each configuration. A median of $0$ occurs for many of the regex matching types. This is because a large number of regexes in the Polyglot corpus do not contain branching or closure operators in their corresponding state machine representations. Thus, many of the regexes do not require any memoization, significantly reducing the median number of memoization entries used.

\begin{figure*}[htbp]
    \centering
    \hfill
    \begin{subfigure}[b]{0.24\textwidth}
        \centering
        \includegraphics[width=\textwidth]{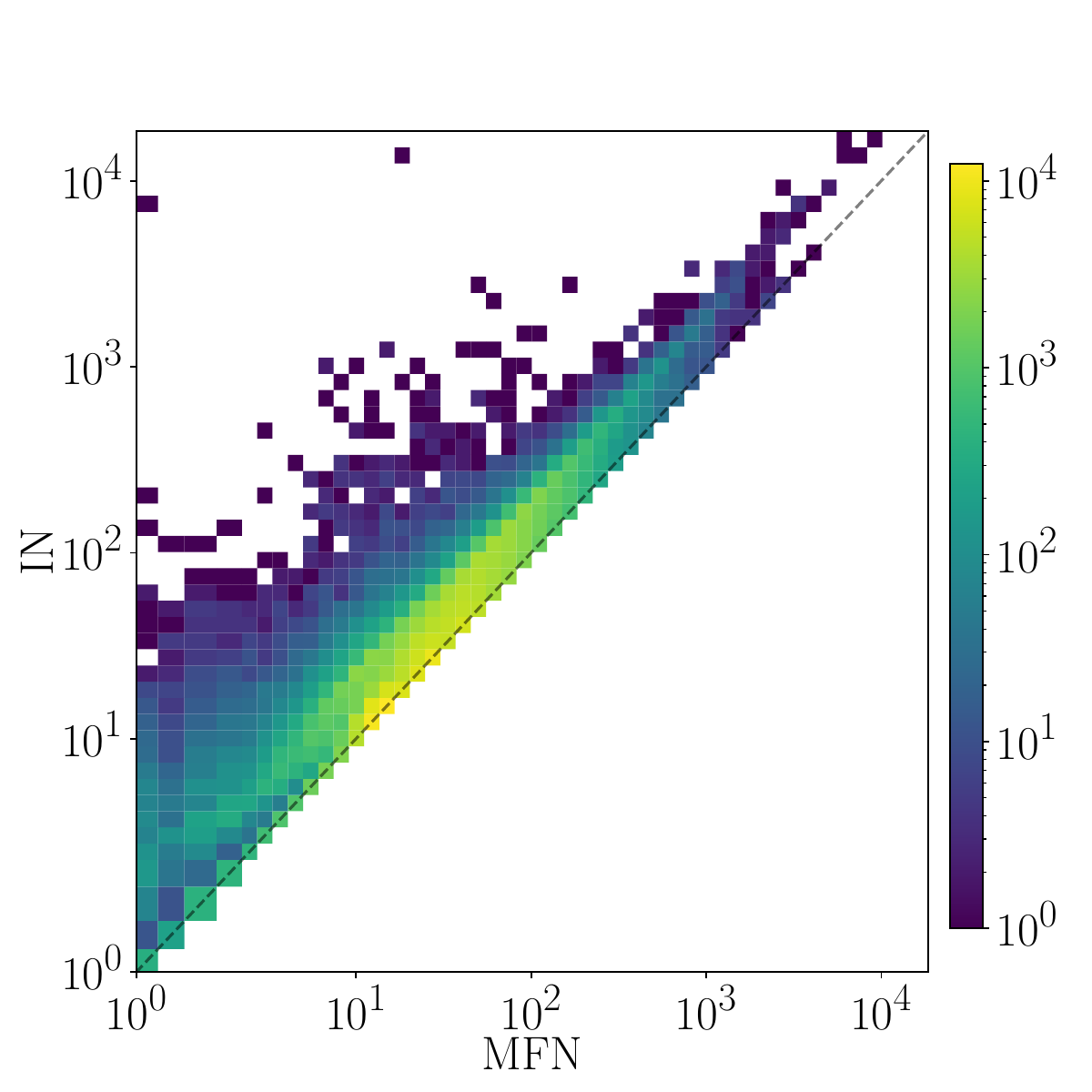}
        \caption{Thompson\\\hspace*{0.55cm} \MFN{} vs \IN{}}
        \label{fig:thompson-mfn-in-size}
    \end{subfigure}
    \hfill
    \begin{subfigure}[b]{0.24\textwidth}
        \centering
        \includegraphics[width=\textwidth]{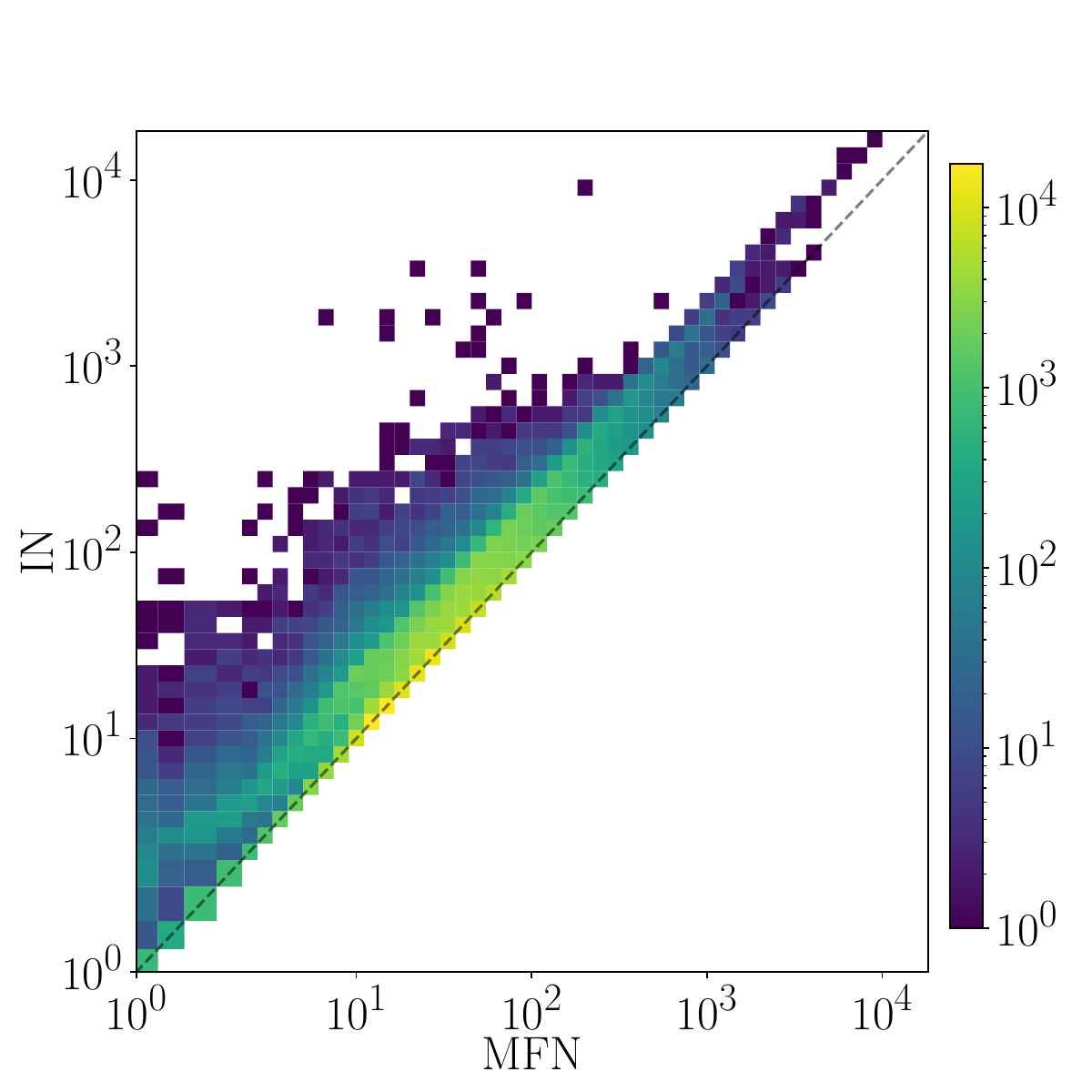}
        \caption{Glushkov\\\hspace*{0.55cm} \MFN{} vs \IN{}}
        \label{fig:glushkov-mfn-in-size}
    \end{subfigure}
    \hfill
    \begin{subfigure}[b]{0.24\textwidth}
        \centering
        \includegraphics[width=\textwidth]{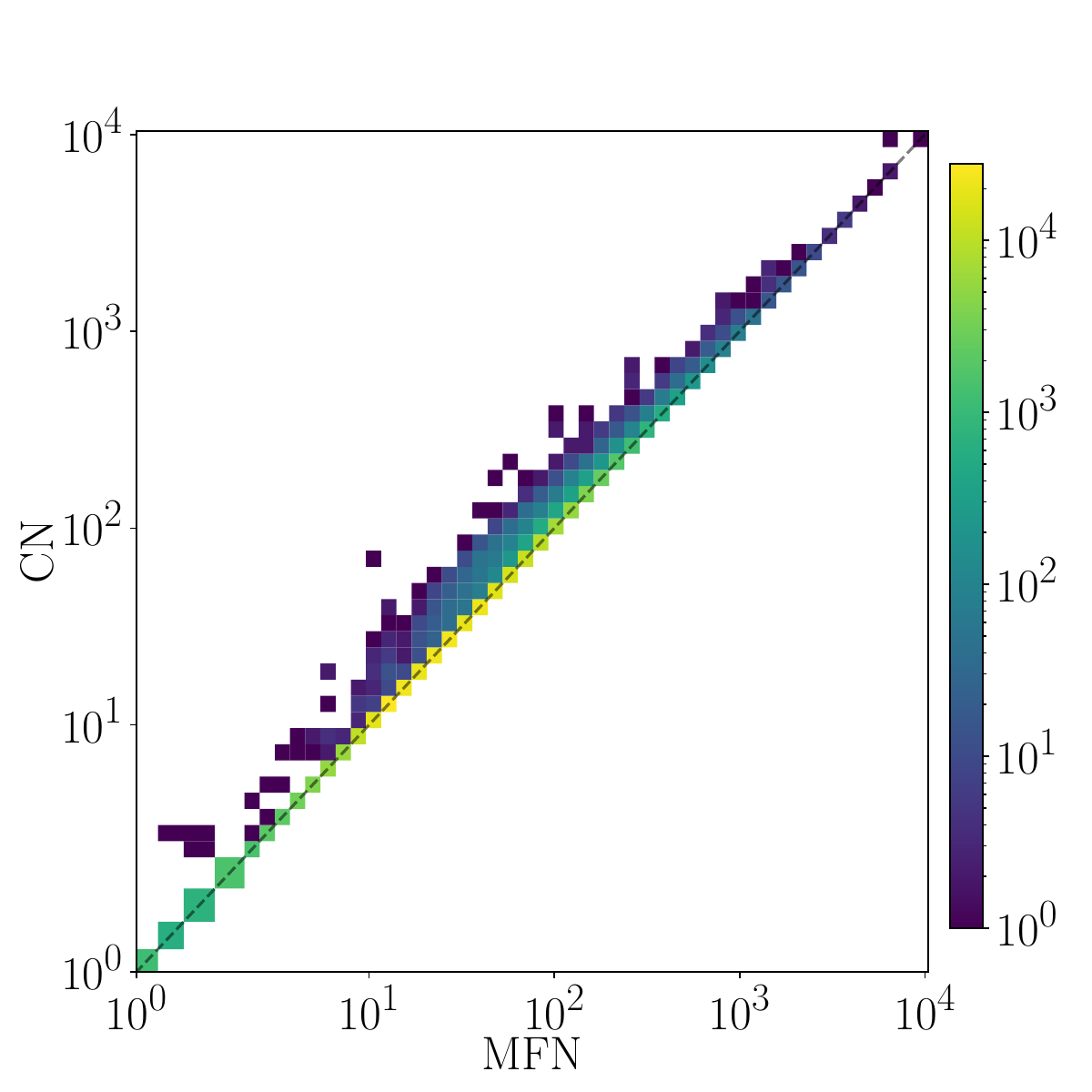}
        \caption{Thompson\\\hspace*{0.55cm} \MFN{} vs \CN{}}
        \label{fig:thompson-mfn-cn-size}
    \end{subfigure}
    \hfill
    \begin{subfigure}[b]{0.24\textwidth}
        \centering
        \includegraphics[width=\textwidth]{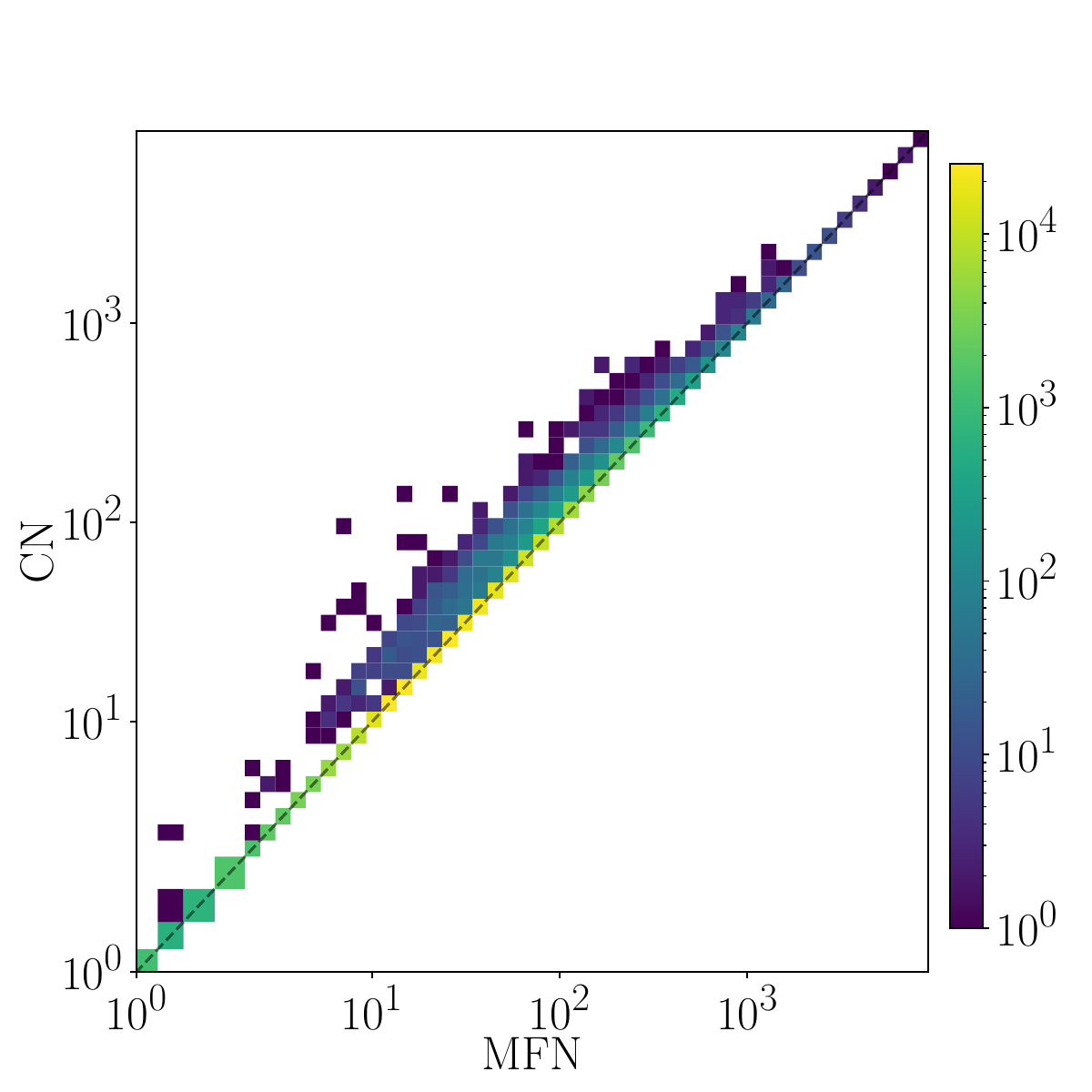}
        \caption{Glushkov\\\hspace*{0.55cm} \MFN{} vs \CN{}}
        \label{fig:glushkov-mfn-cn-size}
    \end{subfigure}
    \hfill\null
    \caption{Comparison of \MFN{} to \IN{} and \CN{} in terms of memoization entries for Thompson and Glushkov constructions. Entries along the diagonal indicate that both schemes give the same result, while the further an entry is into the upper triangle the larger the improvement offered by \MFN{}.}
    \label{fig:memosize}
\end{figure*}

Although the memory improvements of \MFN{} are the main focus of these experiments, matching time is also relevant. \MFN{} removes IDA, but memoizing more states may still improve constant factors in the matching time. A very memory efficient memoization scheme that significantly hinders the matching time would not be very useful in practice.

\MFN{} memoization rarely resulted in a larger number of matching steps compared to \IN{} and \CN{}. Figure~\ref{fig:state} illustrates that \MFN{} memoization was more similar to \IN{} and \CN{} in terms of matching steps than memoization entries. Figure~\ref{fig:state} also shows that there was a case where \MFN{} memoization used an order of magnitude more matching steps compared to \IN{}. However, this was offset by a case where \MFN{} used four orders of magnitude fewer memoization entries compared to \IN{}.

Table~\ref{tab:state} provides the mean, median, and standard deviation of the number of matching steps for each configuration. It supports the fact that \MFN{} incurs only a marginal increase in the average matching steps.

\begin{figure*}[htbp]
    \centering
    \hfill
    \begin{subfigure}[b]{0.24\textwidth}
        \centering
        \includegraphics[width=\textwidth]{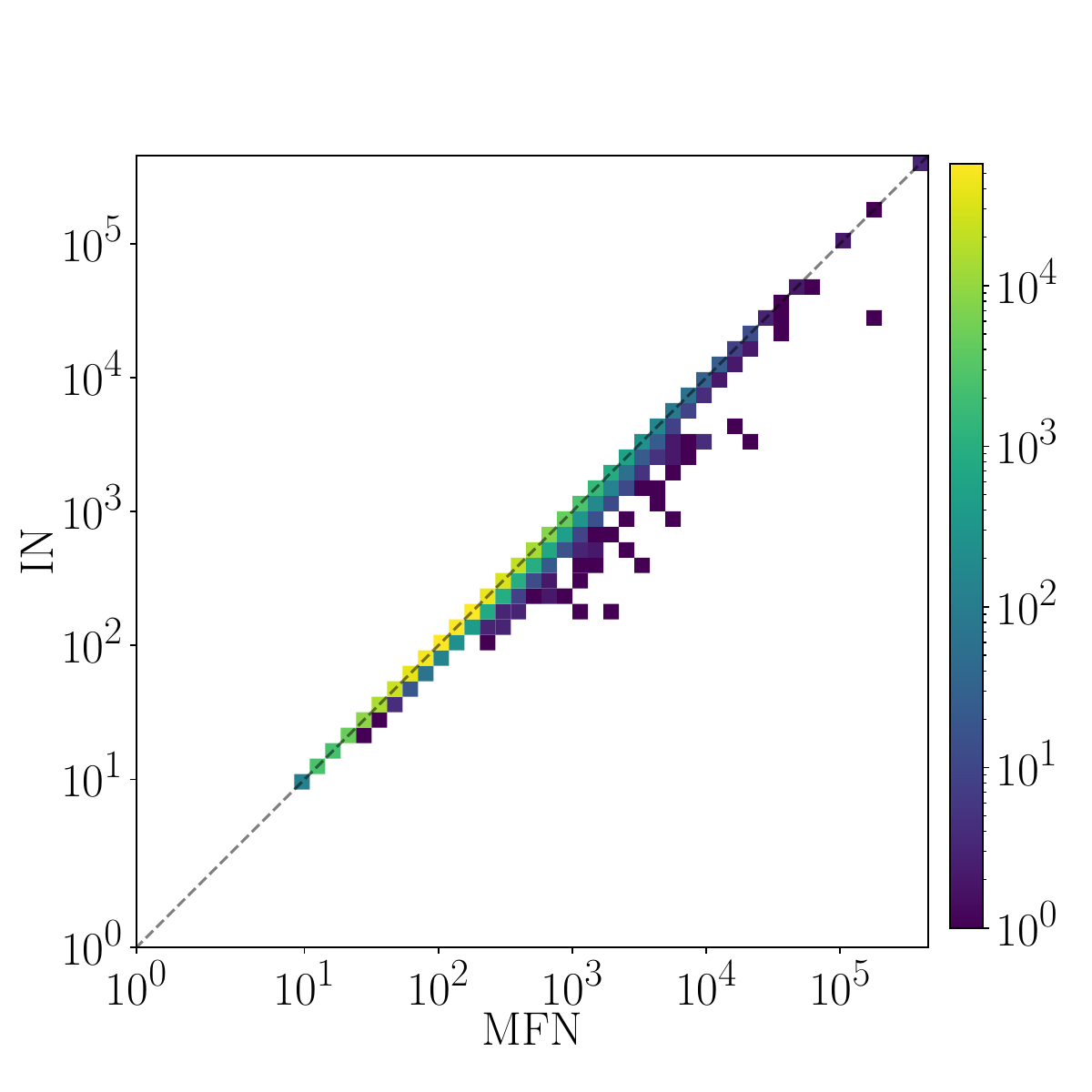}
        \caption{Thompson\\\hspace*{0.55cm} \MFN{} vs \IN{}}
        \label{fig:thompson-mfn-in-size}
    \end{subfigure}
    \hfill
    \begin{subfigure}[b]{0.24\textwidth}
        \centering
        \includegraphics[width=\textwidth]{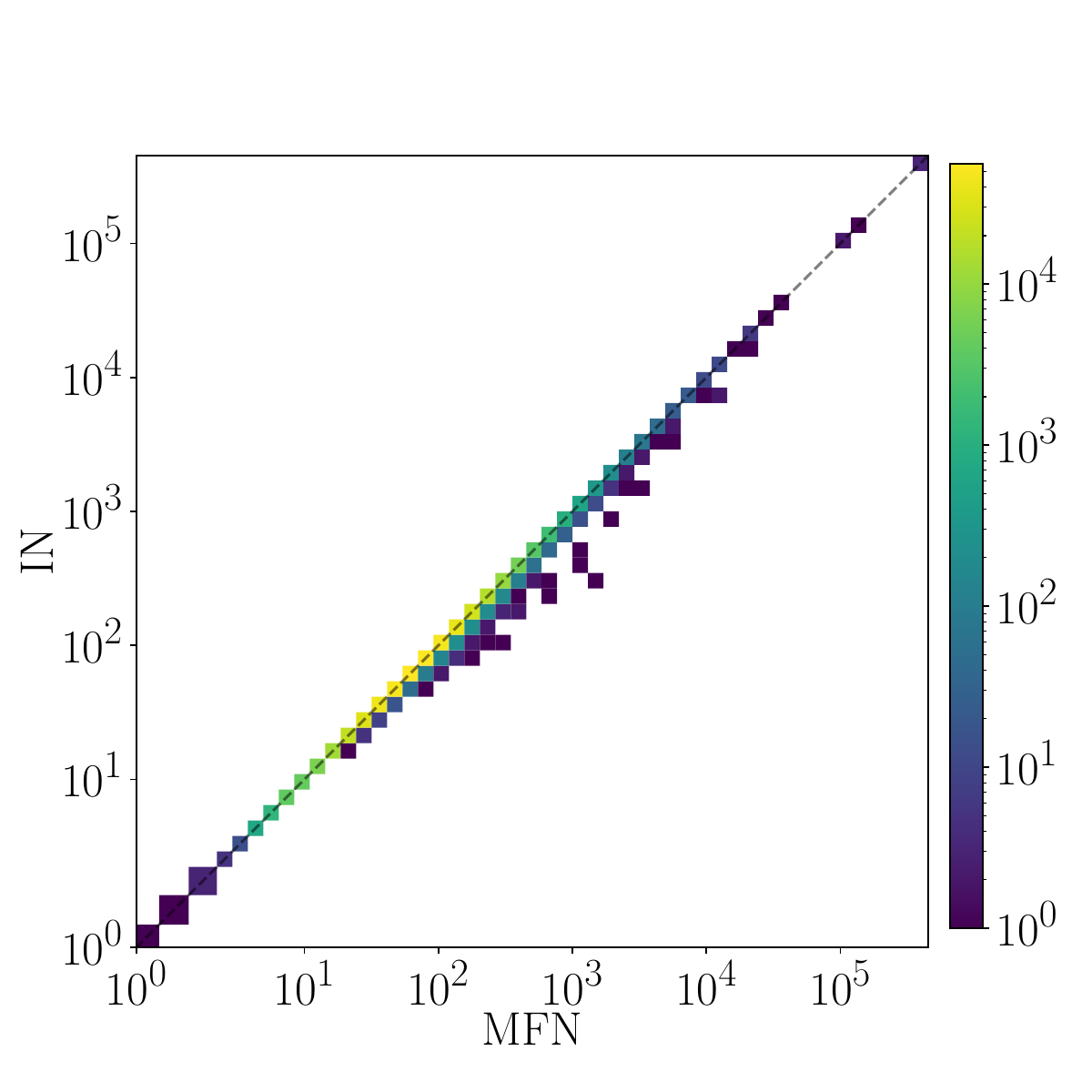}
        \caption{Glushkov\\\hspace*{0.55cm} \MFN{} vs \IN{}}
        \label{fig:glushkov-mfn-in-size}
    \end{subfigure}
    \hfill
    \begin{subfigure}[b]{0.24\textwidth}
        \centering
        \includegraphics[width=\textwidth]{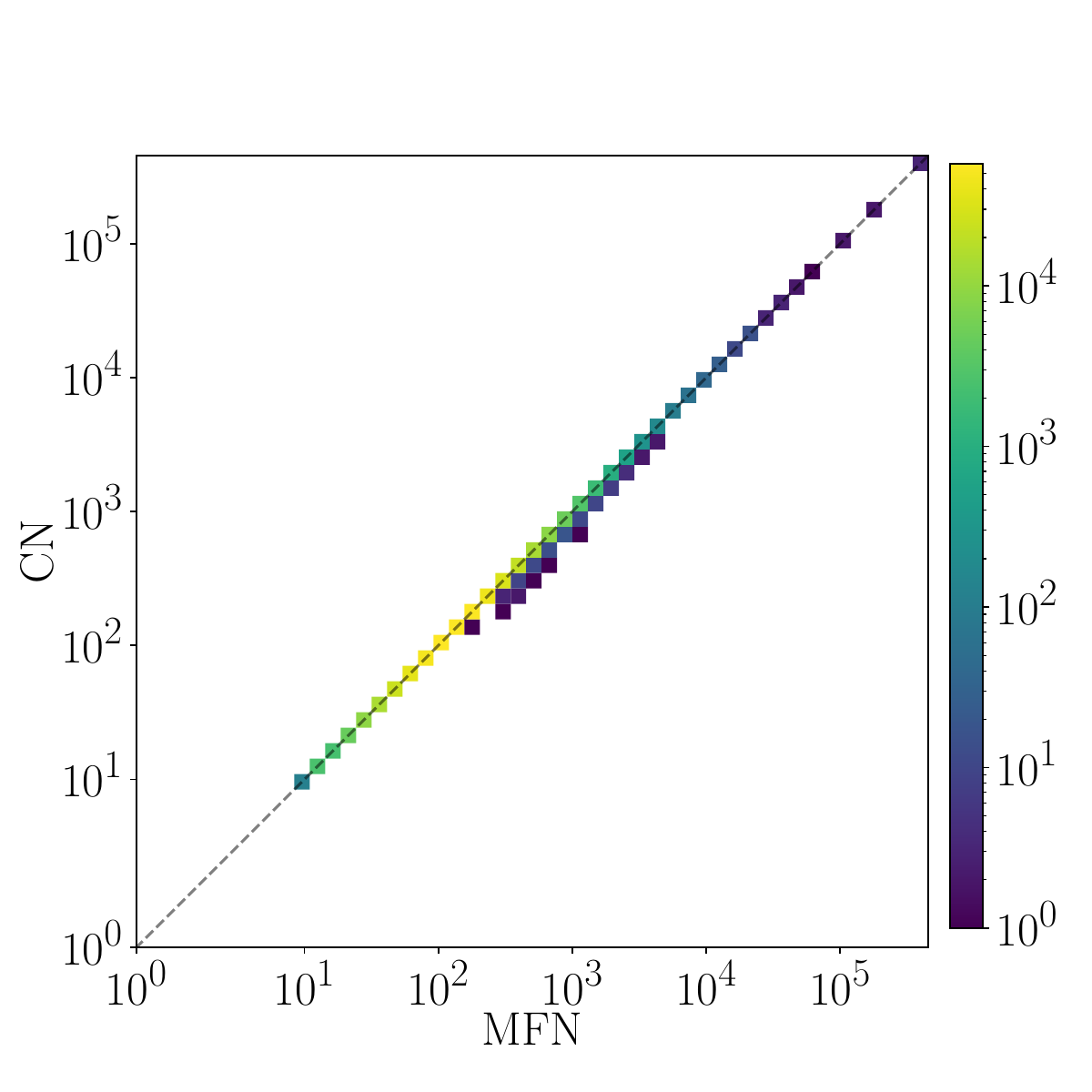}
        \caption{Thompson\\\hspace*{0.55cm} \MFN{} vs \CN{}}
        \label{fig:thompson-mfn-cn-size}
    \end{subfigure}
    \hfill
    \begin{subfigure}[b]{0.24\textwidth}
        \centering
        \includegraphics[width=\textwidth]{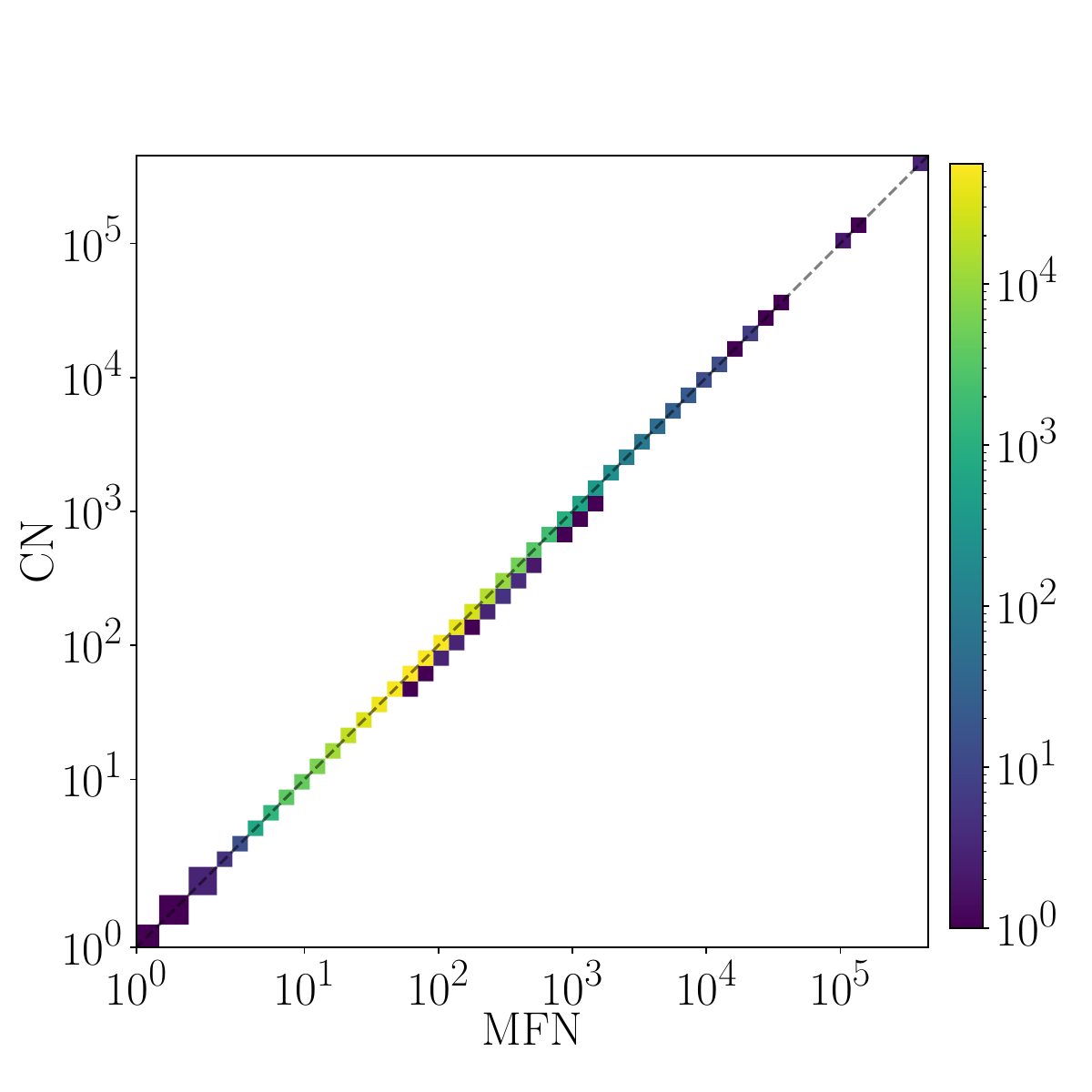}
        \caption{Glushkov\\\hspace*{0.55cm} \MFN{} vs \CN{}}
        \label{fig:glushkov-mfn-cn-size}
    \end{subfigure}
    \hfill\null
    \caption{Comparison of \MFN{} to \IN{} and \CN{} in terms of matching steps for Thompson and Glushkov constructions. Entries along the diagonal correspond to no improvement, as can be seen a scattering of entries fall into the lower triangle, indicating improvements in the number of matching steps performed using \MFN{}.}
    \label{fig:state}
\end{figure*}

Although \MFN{} memoization is not a uniform improvement over all the mentioned memoization schemes, comparison with \CN{} is significant. \MFN{} has a significant memory reduction without a significant increase in matching time. In every case, \MFN{} used the same number of or fewer memoization entries than \CN{}.

Although many cases where \MFN{} reduced the memoization entries resulted in more matching steps, there are many cases where \MFN{} used fewer memoization entries without increasing the matching steps. Conversely, there were no cases where \MFN{} used an equal number of memoization entries and had an increase in the matching steps. %

While MFN resulted in more matching steps than IN and CN in some cases, it still significantly reduced matching steps compared to no memoization. To contrast the difference between MFN and the other memoization schemes, as well as provide insight into the frequency of exponential backtracking in the experiments, Figure~\ref{fig:add-state} illustrates the comparison between no memoization (\None{}) and \MFN{} for both Thompson and Glushkov with respect to the number of matching steps used.

\begin{figure*}[htbp]
    \centering
    \hfill
    \begin{subfigure}[b]{0.49\textwidth}
        \centering
        \includegraphics[width=\textwidth]{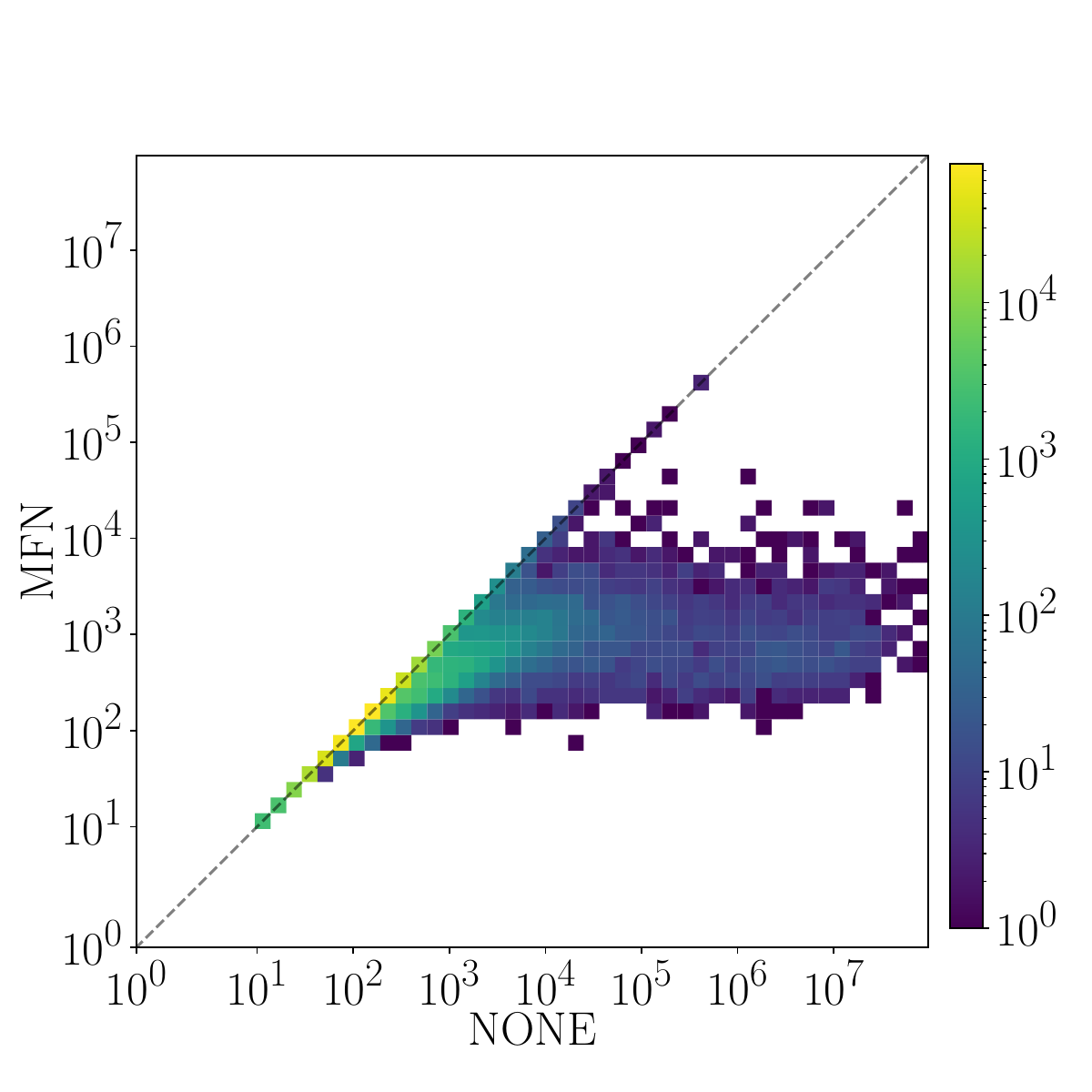}
        \caption{Thompson\\\hspace*{0.55cm} \None{} vs \MFN{}}
        \label{fig:thompson-none-mfn-size}
    \end{subfigure}
    \hfill
    \begin{subfigure}[b]{0.49\textwidth}
        \centering
        \includegraphics[width=\textwidth]{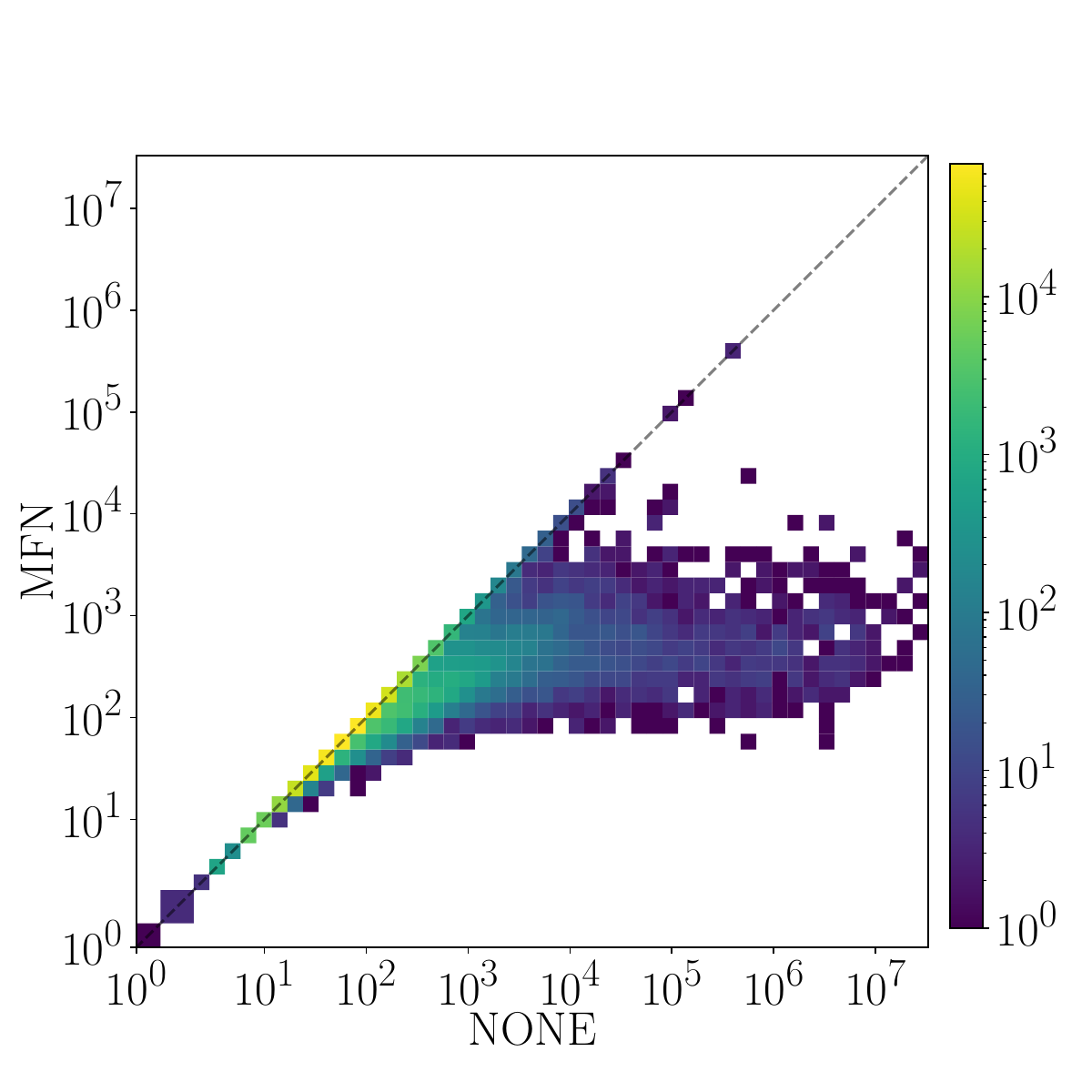}
        \caption{Glushkov\\\hspace*{0.55cm} \None{} vs \MFN{}}
        \label{fig:glushkov-none-mfn-size}
    \end{subfigure}
    \hfill\null
    \caption{Comparison of \None{} to \MFN{} in terms of matching steps for Thompson and Glushkov constructions.}
    \label{fig:add-state}
\end{figure*}

\subsection*{MFN and CN Comparison Examples}

The regex that resulted in one of the worst relative matching time performances of \MFN{} memoization compared to \CN{} memoization for both the Thompson and Glushkov constructions was the following. Observe that this regular expression uses the common practical syntax; refer to the documentation of just about any Perl-style regex matcher (all mentioned in this article are) for details. Consider

\begin{center}
\verb|\w+([-.]\w+)*@\w+([-]\w+)*\.(\w+([-]\w+)*\.)*[a-z]{2,3}$|
\end{center} \noexpand
This regex is the second worst example for Thompson and the fourth worst for Glushkov. In this example, \MFN{} memoization used an average of 5708.1 matching steps, while \CN{} memoization used an average of 5228.1. However, \MFN{} memoization also used an average of 693.2 memoization entries, while \CN{} memoization used an average of 1051.6. For Glushkov, \MFN{} memoization used an average of 2120.6 matching steps, while \CN{} memoization used an average of 2082.6. However, as in the Thompson case, \MFN{} memoization used an average of 693.2 memoization entries, contrasted with 1047.9 for \CN{} memoization. This shows that even in the worst case, \MFN{} memoization is not far worse than \CN{} memoization in terms of matching time.

\subsubsection*{Thompson}

An example of where \MFN{} memoization reduced the number of memoization entries at no cost to matching time for the Thompson construction is given by the following regex.
\begin{center}
\verb|((^[ \t]*>[ \t]?.+\n(.+\n)*\n*)+)|
\end{center} \noexpand
In this example, \MFN{} memoization would memoize four states, while \CN{} memoization would memoize six. This regex has nested closure operators, which allow \MFN{} to take advantage of the Thompson domination property. For this example, \MFN{} memoization used an average of 46.9 and 44.9 memoization entries for full and partial matching, respectively, while \CN{} memoization used an average of 339.9 and 282 memoization entries for full and partial matching, respectively.

\subsubsection*{Glushkov}

An example of where \MFN{} memoization reduced the number of memoization entries at no cost to matching time for the Glushkov construction is the following regex, which is displayed over two lines.

\begin{lstlisting}
^([A-Za-z0-9]+(([-.\_:@+]|--)[A-Za-z0-9]+)*(/([A-Za-z0-9]+(([-.\_:@+]|--)[A-Za-z0-9]+)*))*)$
\end{lstlisting}

In this example, \MFN{} memoization would memoize four states, while \CN{} memoization would memoize nine. This example has nested closure operators, where several apply to subexpressions that contain closure operators that are not nullable. \MFN{} will only memoize states in the deeper subexpressions. For this example, \MFN{} memoization used an average of 132.6 memoization entries for both full and partial matching, while \CN{} memoization used an average of 447.6 memoization entries for both full and partial matching.

\begin{table}[htbp]
    \centering
    \begin{minipage}{0.45\textwidth}
        \centering
        \caption{The mean, median, and standard deviation of string input sizes.}
        \begin{tabular}{|l|c|c|c|}
            \hline
            \textbf{\textit{Type}} & \textbf{\textit{Mean}} & \textbf{\textit{Median}} & \textbf{\textit{Stdev}} \\ 
            \hline
            Positive & 18.94 & 12.0 & 92.49 \\ 
            Negative & 18.99 & 12.0 & 91.76 \\ 
            \hline
            \multicolumn{4}{l}{$^{\mathrm{*}}$All values are rounded to two decimal places.}
            \label{tab:strings}
        \end{tabular}
        \vskip.5cm plus 1cm
        \caption{The mean, median, and standard deviation of the number of memoization entries.}
        \begin{tabular}{|l|c|c|c|}
            \hline
            \multicolumn{4}{|c|}{\textbf{Positive Strings}} \\ 
            \hline
            \textbf{\textit{Type}} & \textbf{\textit{Mean}} & \textbf{\textit{Median}} & \textbf{\textit{Stdev}} \\ 
            \hline
            \multicolumn{4}{|c|}{\textbf{Full -- Glushkov}} \\ 
            \hline
            \IN{} & 10.47 & 2.6 & 33.93 \\ 
            \CN{} & 8.35 & 0.0 & 22.49 \\ 
            \MFN{} & 8.30 & 0.0 & 22.24 \\
            \hline
            \multicolumn{4}{|c|}{\textbf{Full -- Thompson}} \\ 
            \hline
            \IN{} & 11.22 & 4.6 & 35.80 \\ 
            \CN{} & 8.30 & 0.0 & 22.63 \\ 
            \MFN{} & 8.23 & 0.0 & 21.79 \\ 
            \hline
            \multicolumn{4}{|c|}{\textbf{Partial -- Glushkov}} \\ 
            \hline
            \IN{} & 10.60 & 2.7 & 34.01 \\ 
            \CN{} & 8.43 & 0.0 & 22.48 \\ 
            \MFN{} & 8.37 & 0.0 & 22.23 \\ 
            \hline
            \multicolumn{4}{|c|}{\textbf{Partial -- Thompson}} \\ 
            \hline
            \IN{} & 11.35 & 4.7 & 35.91 \\ 
            \CN{} & 8.38 & 0.0 & 22.64 \\ 
            \MFN{} & 8.31 & 0.0 & 21.77 \\ 
            \hline
            \multicolumn{4}{|c|}{\textbf{Negative Strings}} \\ 
            \hline
            \textbf{\textit{Type}} & \textbf{\textit{Mean}} & \textbf{\textit{Median}} & \textbf{\textit{Stdev}} \\ 
            \hline
            \multicolumn{4}{|c|}{\textbf{Full -- Glushkov}} \\ 
            \hline
            \IN{} & 8.31 & 0.0 & 32.58 \\ 
            \CN{} & 5.49 & 0.0 & 19.76 \\ 
            \MFN{} & 5.38 & 0.0 & 19.16 \\ 
            \hline
            \multicolumn{4}{|c|}{\textbf{Full -- Thompson}} \\ 
            \hline
            \IN{} & 9.16 & 0.6 & 36.18 \\ 
            \CN{} & 5.41 & 0.0 & 19.52 \\ 
            \MFN{} & 5.29 & 0.0 & 18.80 \\ 
            \hline
            \multicolumn{4}{|c|}{\textbf{Partial -- Glushkov}} \\ 
            \hline
            \IN{} & 9.27 & 0.0 & 36.01 \\ 
            \CN{} & 6.25 & 0.0 & 22.50 \\ 
            \MFN{} & 6.15 & 0.0 & 21.94 \\ 
            \hline
            \multicolumn{4}{|c|}{\textbf{Partial -- Thompson}} \\ 
            \hline
            \IN{} & 10.04 & 0.6 & 41.13 \\ 
            \CN{} & 6.19 & 0.0 & 22.40 \\ 
            \MFN{} & 6.07 & 0.0 & 21.49 \\ 
            \hline
            \multicolumn{4}{l}{$^{\mathrm{*}}$All values are rounded to two decimal places.}
            \label{tab:size}
        \end{tabular}
    \end{minipage}
    \hfill %
    \begin{minipage}{0.45\textwidth}
        \centering
        \caption{The mean, median, and standard deviation of the number of matching steps.}
        \begin{tabular}{|l|c|c|c|}
            \hline
            \multicolumn{4}{|c|}{\textbf{Positive Strings}} \\ 
            \hline
            \textbf{\textit{Type}} & \textbf{\textit{Mean}} & \textbf{\textit{Median}} & \textbf{\textit{Stdev}} \\ 
            \hline
            \multicolumn{4}{|c|}{\textbf{Full -- Glushkov}} \\ 
            \hline
            \IN{} & 24.77 & 14.2 & 102.06 \\ 
            \CN{} & 24.78 & 14.2 & 102.08 \\ 
            \MFN{} & 24.78 & 14.2 & 102.08 \\ 
            \None{} & 77.05 & 14.2 & 17914.10 \\ 
            \hline
            \multicolumn{4}{|c|}{\textbf{Full -- Thompson}} \\ 
            \hline
            \IN{} & 44.89 & 26.0 & 130.11 \\ 
            \CN{} & 44.95 & 26.0 & 132.79 \\ 
            \MFN{} & 44.95 & 26.0 & 132.79 \\ 
            \None{} & 395.16 & 26.0 & 80200.51 \\ 
            \hline
            \multicolumn{4}{|c|}{\textbf{Partial -- Glushkov}} \\ 
            \hline
            \IN{} & 25.74 & 15.2 & 102.41 \\ 
            \CN{} & 25.75 & 15.2 & 102.43 \\ 
            \MFN{} & 25.75 & 15.2 & 102.43 \\ 
            \None{} & 78.16 & 15.2 & 17914.03 \\ 
            \hline
            \multicolumn{4}{|c|}{\textbf{Partial -- Thompson}} \\ 
            \hline
            \IN{} & 43.69 & 25.0 & 131.28 \\ 
            \CN{} & 43.76 & 25.0 & 133.98 \\ 
            \MFN{} & 43.76 & 25.0 & 133.99 \\ 
            \None{} & 358.79 & 25.0 & 76662.17 \\ 
            \hline
            \multicolumn{4}{|c|}{\textbf{Negative Strings}} \\ 
            \hline
            \textbf{\textit{Type}} & \textbf{\textit{Mean}} & \textbf{\textit{Median}} & \textbf{\textit{Stdev}} \\ 
            \hline
            \multicolumn{4}{|c|}{\textbf{Full -- Glushkov}} \\ 
            \hline
            \IN{} & 17.96 & 8.9 & 63.07 \\ 
            \CN{} & 17.99 & 8.9 & 63.47 \\ 
            \MFN{} & 18.00 & 8.9 & 63.47 \\ 
            \None{} & 1316.44 & 9.0 & 79361.50 \\ 
            \hline
            \multicolumn{4}{|c|}{\textbf{Full -- Thompson}} \\ 
            \hline
            \IN{} & 54.54 & 29.8 & 171.27 \\ 
            \CN{} & 55.21 & 29.8 & 175.91 \\ 
            \MFN{} & 55.23 & 29.8 & 176.09 \\ 
            \None{} & 7872.27 & 29.9 & 290233.72 \\ 
            \hline
            \multicolumn{4}{|c|}{\textbf{Partial -- Glushkov}} \\ 
            \hline
            \IN{} & 36.59 & 17.1 & 1074.10 \\ 
            \CN{} & 36.78 & 17.1 & 1074.32 \\ 
            \MFN{} & 36.78 & 17.1 & 1074.32 \\ 
            \None{} & 981.30 & 17.4 & 74635.75 \\ 
            \hline
            \multicolumn{4}{|c|}{\textbf{Partial -- Thompson}} \\ 
            \hline
            \IN{} & 62.95 & 29.9 & 1083.54 \\ 
            \CN{} & 64.84 & 30.0 & 1109.72 \\ 
            \MFN{} & 64.86 & 30.0 & 1109.75 \\  
            \None{} & 5596.49 & 30.2 & 292893.14 \\ 
            \hline
            \multicolumn{4}{l}{$^{\mathrm{*}}$All values are rounded to two decimal places.}
            \label{tab:state}
        \end{tabular}
    \end{minipage}
\end{table}

\section{Conclusion and Future Work}
\label{sec:conclusion}

We introduced the \emph{Minimum Feedback Node} (\MFN{}) memoization
scheme for backtracking regular expression matchers, based on computing
a minimum feedback vertex set for an NFA. We presented linear-time parse
tree algorithms for both Thompson and Glushkov NFA constructions.
We showed through examples that \I{} starting from \CN{} does not always produce a minimum memoization set -- a finding that challenges the conjecture of~\cite{van2021memoized} (at least for the Glushkov construction) -- and that \MFN{} may memoize strictly fewer states than \I{}.
We further showed that Thompson NFAs are always at least as selective
as Glushkov NFAs under any of the studied memoization schemes, providing
a theoretical justification for using Thompson constructions when aiming to reduce the number of states
being memoized. Our experimental evaluation confirms
that \MFN{} reduces memory usage compared to \CN{} and \IN{} with negligible
impact on matching time.
Two natural follow-up directions deserve mention. First, we did not include \I{} in our corpus-scale experiments. The schemes we compared -- \None{}, \IN{}, \CN{}, and \MFN{} -- all admit memoization-set computation in time linear in $|R|$, the size of the regular expression, and our experimental setup was designed around this common cost model. \I{} does not fit this model: as defined in~\cite{van2021memoized}, it iterates over the states of \CN{} in some order and performs an IDA check at each step, with a naive implementation running in $O(|Q|^4)$ time. More fundamentally, the resulting set is \emph{order-dependent}, so a meaningful comparison must either fix a canonical order 
or quantify the behaviour across all orders. 
Second, the corpus-level averages we report are dominated by regexes
without IDA, for which all selective schemes use few entries, and the
absolute gap between them is small. 
A focused experimental study restricted to such regexes, drawing on existing ReDoS benchmarks in addition to the catastrophic subset of the Polyglot corpus, is left as future work. 

Future work includes: (i) proving or disproving Conjecture~\ref{conj:mfn-contains-min} or~\ref{conj:mfn-refine}; (ii) determining the precise complexity of computing a minimum memoization set for Thompson NFAs --- while the NP-hardness result of~\cite{van2021memoized} applies to general NFAs, the restricted structure of Thompson NFAs suggests that polynomial-time algorithms may still be possible; (iii) extending \MFN{} to prioritized NFAs and studying its interaction with extended regex features such as counters, backreferences, and lookaheads; and (iv) developing memoization strategies for regexes with counters
(at present, counters are handled by expansion, but this can obscure the trade-off between memoization and matching time: for example, with \CN{}, the regex $(a\Or a)\{1,100\}$ yields no memoization yet performs poorly on inputs of the form $aa\ldots ab$, whereas \IN{} handles this example well but at the cost of substantial memoization); and (v) a dedicated empirical comparison of MFN with \I{}, including an analysis of the sensitivity of IAR to the ordering of states in \CN{}, conducted on the catastrophic subset of the Polyglot corpus.

\bibliographystyle{splncs04}
\bibliography{main}

\begin{thebibliography}{10}
\providecommand{\url}[1]{\texttt{#1}}
\providecommand{\urlprefix}{URL }
\providecommand{\doi}[1]{https://doi.org/#1}

\bibitem{inthewild}
Berglund, M., van~der Merwe, B.: On the semantics of regular expression parsing
  in the wild. Theoretical Computer Science  \textbf{679},  69--82 (2017).
  \doi{10.1016/j.tcs.2016.09.006}

\bibitem{glushkov-tcs}
Berry, G., Sethi, R.: From regular expressions to deterministic automata.
  Theor. Comput. Sci.  \textbf{48}(3),  117--126 (1986).
  \doi{10.1016/0304-3975(86)90088-5}

\bibitem{davis2019aren}
Davis, J.C., IV, L.G.M., Coghlan, C.A., Servant, F., Lee, D.: Why aren't
  regular expressions a lingua franca? {An} empirical study on the re-use and
  portability of regular expressions. In: Dumas, M., Pfahl, D., Apel, S.,
  Russo, A. (eds.) Proceedings of the {ACM} Joint Meeting on European Software
  Engineering Conference and Symposium on the Foundations of Software
  Engineering, {ESEC/SIGSOFT} {FSE} 2019, Tallinn, Estonia, August 26-30, 2019.
  pp. 443--454. {ACM} (2019). \doi{10.1145/3338906.3338909}

\bibitem{davis2021using}
Davis, J.C., Servant, F., Lee, D.: {Using} {Selective} {Memoization} to
  {Defeat} {Regular} {Expression} {Denial} of {Service} {(ReDoS)}. In: 42nd
  {IEEE} Symposium on Security and Privacy, {SP} 2021, San Francisco, CA, USA,
  24-27 May 2021. pp. 1--17. {IEEE} (2021). \doi{10.1109/SP40001.2021.00032}

\bibitem{GiammarresiPontyWood1998}
Giammarresi, D., Ponty, J.L., Wood, D.: {Glushkov} and {Thompson}
  {Constructions}: {A} {Synthesis}. Technical Report 98-17, Universit{\`a} Ca'
  Foscari di Venezia (1998)

\bibitem{giammarresi2004thompson}
Giammarresi, D., Ponty, J., Wood, D., Ziadi, D.: A characterization of
  {Thompson} digraphs. Discret. Appl. Math.  \textbf{134}(1-3),  317--337
  (2004). \doi{10.1016/S0166-218X(03)00299-3}

\bibitem{glushkov1961abstract}
Glushkov, V.M.: The abstract theory of automata. Russian Mathematical Surveys
  \textbf{16}(5), ~1 (oct 1961). \doi{10.1070/RM1961v016n05ABEH004112}

\bibitem{karp1972}
Karp, R.M.: Reducibility among combinatorial problems. In: Miller, R.E.,
  Thatcher, J.W. (eds.) Proceedings of a symposium on the Complexity of
  Computer Computations, held March 20-22, 1972, at the {IBM} Thomas J. Watson
  Research Center, Yorktown Heights, New York, {USA}. pp. 85--103. The {IBM}
  Research Symposia Series, Plenum Press, New York (1972).
  \doi{10.1007/978-1-4684-2001-2\_9}

\bibitem{van2021memoized}
van~der Merwe, B., Mouton, J., van Litsenborgh, S., Berglund, M.: Memoized
  regular expressions. In: Maneth, S. (ed.) Implementation and Application of
  Automata - 25th International Conference, {CIAA} 2021, Virtual Event, July
  19-22, 2021, Proceedings. Lecture Notes in Computer Science, vol. 12803, pp.
  39--52. Springer (2021). \doi{10.1007/978-3-030-79121-6\_4}

\bibitem{roodt2024benchmarking}
Roodt, A., Watling, B.K.M., Bester, W., van~der Merwe, B., Sung, S., Han, Y.S.:
  Benchmarking regular expression matching. In: Fazekas, S.Z. (ed.)
  Implementation and Application of Automata - 28th International Conference,
  {CIAA} 2024, Akita, Japan, September 3-6, 2024, Proceedings. Lecture Notes in
  Computer Science, vol. 15015, pp. 316--331. Springer (2024).
  \doi{10.1007/978-3-031-71112-1\_23}

\bibitem{shamir1979linear}
Shamir, A.: A linear time algorithm for finding minimum cutsets in reducible
  graphs. {SIAM} J. Comput.  \textbf{8}(4),  645--655 (1979).
  \doi{10.1137/0208051}

\bibitem{thompson1968}
Thompson, K.: Regular expression search algorithm. Commun. {ACM}
  \textbf{11}(6),  419--422 (1968). \doi{10.1145/363347.363387}

\bibitem{weber1991degree}
Weber, A., Seidl, H.: On the degree of ambiguity of finite automata.
  {Theoretical} {Computer} {Science}  \textbf{88}(2),  325--349 (1991).
  \doi{10.1016/0304-3975(91)90381-B}

\bibitem{weideman2016analyzing}
Weideman, N., van~der Merwe, B., Berglund, M., Watson, B.W.: Analyzing matching
  time behavior of backtracking regular expression matchers by using ambiguity
  of {NFA}. In: Han, Y., Salomaa, K. (eds.) Implementation and Application of
  Automata - 21st International Conference, {CIAA} 2016, Seoul, South Korea,
  July 19-22, 2016, Proceedings. Lecture Notes in Computer Science, vol.~9705,
  pp. 322--334. Springer (2016). \doi{10.1007/978-3-319-40946-7\_27}

\end{thebibliography}

\end{document}